\documentclass[11pt]{article} 

\usepackage{times}
\usepackage{natbib}

\usepackage{fullpage}
\usepackage{booktabs}       
\usepackage{amsfonts}       
\usepackage{nicefrac}       
\usepackage{microtype}
\usepackage[linesnumbered,ruled,algo2e]{algorithm2e}

\usepackage{amsmath}
\usepackage{amsthm}
\usepackage{amsfonts}
\usepackage{amssymb}
\usepackage{amstext}
\usepackage{xspace}
\usepackage{graphicx}
\usepackage{enumitem}
\usepackage[pagebackref,colorlinks,linkcolor=blue,filecolor = blue,citecolor = blue, urlcolor = blue, hyperfootnotes=false, breaklinks]{hyperref}

\usepackage{bm}
\usepackage{environ}
\usepackage{thmtools}
\usepackage{thm-restate}
\usepackage[capitalize, noabbrev, nameinlink]{cleveref}
\usepackage{enumitem}

\newcommand{\prodp}{\mbox{$\prod_{i=1}^nP_i$}}
\newcommand{\prodq}{\mbox{$\prod_{i=1}^nQ_i$}}
\newcommand{\bern}{\mathrm{Bern}}
\newcommand{\ber}{\mathrm{Bern}}
\newcommand{\poi}{\mathrm{Poi}}

\newcommand{\dtv}{\mbox{${d}_{\mathrm{TV}}$}}
\newcommand{\dtvsqr}{\mbox{${d}_{\mathrm{TV}}^2$}}
\newcommand{\dhel}{\mbox{$d_{\mathrm H}$}} 
\newcommand{\dhelsqr}{\mbox{$d^2_{\mathrm H}$}} 
\newcommand{\dkl}{\mbox{$d_{\mathrm {KL}}$}} 
\newcommand{\dchisqr}{\mbox{$d_{\mathrm \chi^2}$}} 
\newcommand{\e}{\mathrm{E}}
\newcommand{\var}{\mathrm{Var}} 
\newcommand{\cP}{f(P)}
\newcommand{\cQ}{f(Q)}

\declaretheorem[name=Theorem, within=section]{theoremRe}
\declaretheorem[name=Corollary, sibling=theoremRe]{corollaryRe}

\declaretheorem[name=Lemma, sibling=theoremRe]{lemmaRe}
\declaretheorem[name=Claim, sibling=theoremRe]{claimRe}

\declaretheorem[name=Definition, sibling=theoremRe]{definitionRe}

\declaretheorem[name=Fact, sibling=theoremRe]{factRe}

\newcommand{\eps}{\epsilon}
\newcommand{\ignore}[1]{}

\newcommand{\comment}[1]{}

\let\ab\allowbreak
\usepackage[utf8]{inputenc}

\allowdisplaybreaks

\title{Testing Product Distributions: A Closer Look}
\author{%
 Arnab Bhattacharyya\thanks{Supported in part by the MOE2019-T2-1-152 grant.}\\\texttt{arnabb@nus.edu.sg}\\National University of Singapore
 \and
 Sutanu Gayen\thanks{Supported in part by NRF-AI Fellowship R-252-100-B13-281.}\\\texttt{sutanugayen@gmail.com}\\National University of Singapore
 \and
 Saravanan Kandasamy\\\texttt{sk3277@cornell.edu}\\Cornell University
 \and
 N. V. Vinodchandran\thanks{Supported in part by the US National Science
Foundation under grants NSF CCF-184908 and NSF HDR:TRIPODS-1934884.}\\\texttt{vinod@cse.unl.edu}\\University of Nebraska-Lincoln
}

\begin{document}

\maketitle
\begin{abstract}
We study the problems of {\em identity} and {\em closeness testing} of $n$-dimensional product distributions. Prior works by Canonne, Diakonikolas, Kane and Stewart (COLT `17) and Daskalakis and Pan (COLT `17) have established tight sample complexity bounds for {\em non-tolerant testing over a binary alphabet}: given two product distributions $P$ and $Q$ over a binary alphabet, distinguish between the cases $P=Q$ and $\dtv(P,Q)>\eps$. We build on this prior work to give a more comprehensive map of the complexity of testing of product distributions by investigating {\em tolerant testing  with respect to several natural distance measures and over an arbitrary alphabet}. \ignore{In particular,  we design efficient algorithms for tolerant identity and closeness testing. We also show nearly matching lower bounds on the sample complexity of many of these problems with respect to $n$.} Our study gives a fine-grained understanding of how the sample complexity of tolerant testing varies with the distance measures  for product distributions. In addition, we also extend one of our upper bounds on product distributions to bounded-degree Bayes nets.
\end{abstract}


\newpage

\section{Introduction}
The main goal of this work is to give a comprehensive investigation to the sample complexity of several distribution testing problems over {\em high-dimensional product distributions}. 
Testing properties of distributions from samples has been actively investigated for several decades 
from the perspectives of classical statistics and, more recently, from a property testing viewpoint in theoretical computer science. Hypothesis testing is a classical problem investigated in statistics with significant practical applications.  From the property testing viewpoint, the two most well studied distribution testing problems are {\em identity testing} and {\em closeness testing.}\footnote{Identity testing is also known as {\em goodness-of-fit testing} or {\em one-sample testing} in the literature. Similarly, closeness testing is also known as {\em two-sample testing.}}

In  the {\em identity testing problem},  we are given a known reference distribution $Q$ and sample access to an unknown distribution $P$ over the same sample space as that of $Q$, and the goal is to distinguish between the cases $P = Q$ or $P$ is 
$\epsilon$-far  from $Q$ with respect to a given distance measure. It is known that $\Theta(\sqrt{m}/\epsilon^2)$ samples are necessary and sufficient to solve the identity testing problem with respect to the total variation distance, where $m$ is the size of the sample space  \citep*{Valiant:2014:AIP:2706700.2707449,Pan08}.  In the {\em closeness testing problem}, we have sample access to a pair of unknown distributions $P$ and $Q$ on a common sample space, and the goal is to distinguish between the cases $P=Q$ or  $P$ is $\eps$-far from $Q$ with respect to a certain distance measure.
It is known that $\Theta(\max(m^{2/3}\eps^{-4/3}, \sqrt{m}\eps^{-2}))$ samples are necessary and sufficient to solve the closeness testing problem with respect to the total variation distance, where $m$ is the size of the sample space \citep*{ChanDVV14}. 
See the surveys \citet*{DBLP:journals/crossroads/Rubinfeld12, DBLP:journals/eccc/Canonne15} and the references therein for pointers to the extensive research on identity and closeness testing as well as related problems. 

\comment{ Vinod's Text
In the {\em tolerant} version of the above problems, which is arguably more practical, instead of testing for $P=Q$ vs $P$ and $Q$ are  $\epsilon$-far, the goal is to distinguish between the cases $P$ and $Q$ are $\epsilon$-close versus $P$ and $Q$ are $3\epsilon$ far. 
}

One of the main bottlenecks resulting from the above-mentioned  complexity bounds is that these testing problems are (provably) hard for arbitrary distributions over large sample spaces.  
For example, for distributions over an $n$-dimensional Boolean hypercube $m=2^n$ and hence $\Theta(2^{n\over 2})$ samples are necessary and sufficient for identity testing (for a constant $\epsilon$). 
To overcome this bottleneck, very recently researchers have started investigating testing problems over high-dimensional sample spaces by imposing natural structural assumptions over distributions. Such assumptions restrict the class of distributions and open up the possibility of designing testers with substantially smaller sample complexity than required for the the general case. Among them {\em product distributions}
over a finite alphabet are a natural class that is both practically relevant and simple enough to serve as a test ground for algorithm design. Indeed, prior works of \citet*{DBLP:conf/colt/CanonneDKS17} and ~\citet*{DBLP:conf/colt/DaskalakisP17} have established tight sample complexity bounds for identity and closeness testing of product distributions over a binary alphabet. 

A drawback of the testing problems  as stated is their one-sided or {\em non-tolerant} aspect: on the one side of the decision, we only need to distinguish from the case where two distributions are {\em exactly equal}.  This is a significant restriction specially  for high-dimensional distributions which require a large number of parameters to be specified. For example, in the case of identity testing, it is unlikely that we can ever hypothesize a reference distribution $Q$ such that it exactly equals the data distribution $P$. Similarly, for closeness testing, two data distributions $P$ and $Q$ are most likely not exactly equal. The {\em tolerant} version of testing problems 
addresses this issue  as it seeks to design  testers for identity and closeness that {\em tolerate} errors on both decision cases. That is, in the tolerant version we would like to distinguish between the cases $\dtv(P,Q)\leq \eps_1$ and $\dtv(P,Q)>\eps_2$, where $\eps_1 < \eps_2$ are user-supplied error parameters. The tolerance requirement makes the testing problems more expensive. For arbitrary distributions on a set of size $m$ it is known \citep*{ValiantV10} that  tolerant identity and closeness testing of arbitrary distributions supported on a set of size $m$ require $\Omega(m/\log m)$ samples for constants $\eps_1< \eps_2$, if the distance measure used is the total variation distance. 

The main focus of this paper is to take a closer look at the complexity of testing of product distributions by investigating {\em tolerant testing  with respect to several natural distance measures and over an arbitrary alphabet}. Such an investigation is important because a complete picture on the complexity of testing product distribution will shed light on possibilities and challenges in algorithm design for testing high dimensional structured distributions.

\ignore{In this paper, we study {\em tolerant testing} over a {\em high-dimensional} domain. In tolerant testing, the goal is to distinguish pairs of distributions which are $\eps_1$-close to each other from pairs which are $\eps_2$-far from each other (with respect to some distance), for given $0 < \eps_1 < \eps_2 < 1$. Thus, one can define the tolerant identity testing and the tolerant closeness testing problems. It is clear that in most applications, one would like identity and closeness tests to be tolerant because usually, the two input distributions $P$ and $Q$ come from different sources and are unlikely to be exactly equal in the case when they are not far from each other. We also focus on the setting where the domain consists of elements in $\Sigma^n$, i.e., strings of length $n$ over an alphabet $\Sigma$. This setting is particularly prevalent in modern machine learning and statistical applications in which each sample corresponds to a high-dimensional feature vector.

  Additionally, in our setting, the domain size $m$ equals $|\Sigma|^n$, which makes the sample complexity $\Omega(|\Sigma|^n/n\log{|\Sigma|})$. }

\subsection{Our Contributions}

We investigate tolerant testing of product distributions with respect to the  following distance measures:  total variation distance ($\dtv$), Hellinger distance ($\dhel$), Kullback-Leibler divergence ($\dkl$), and Chi-squared distance ($\dchisqr$)\footnote{Refer to \cref{prelims} for the notations and definitions.}. The following relationship is well-known among them: 
\begin{align}\label{eqn:Reldistances}
\dhelsqr(P,Q) \leq \dtv(P,Q) \leq \sqrt{2} \dhel(P,Q) \leq \sqrt{\dkl(P,Q)} \leq \sqrt{\dchisqr(P,Q)}
\end{align}
We fix a pair of distance functions $d_1 \le d_2$ from the above equation and investigate the problem of deciding $d_1(P,Q)\le \epsilon/3$ versus $d_2(P,Q) >\epsilon$ with 2/3 probability, which we call $d_1$-versus-$d_2$ testing. When both $P$ and $Q$ are only accessed by samples, this problem is called $d_1$-versus-$d_2$ closeness testing. When $Q$ is a reference distribution given to us and $P$ is accessed by samples, the problem is called $d_1$-versus-$d_2$ identity testing. The problem of distinguishing $P=Q$ versus $d_2(P,Q) > \epsilon$ is called non-tolerant testing w.r.t. $d_2$. Clearly, tolerant testing is at least as hard as non-tolerant testing.

Our contributions regarding $d_1$-versus-$d_2$ identity and closing testing problems over product distributions are summarized in \cref{table:identity} and \cref{table:closeness}. Each cell of the tables represents the sample complexity of testing whether the two product distributions are close or far in terms of the distance corresponding to that row and column respectively. The problems become harder as we traverse the table down or to the right due to \cref{eqn:Reldistances}. \citet*{Daskalakis:2018:DDS:3174304.3175479} have shown that non-tolerant testing w.r.t. $\dkl$ is not testable in a finite set of samples. Hence, only $\dtv$ and $\dhel$ are meaningful for $d_2$. 

\begin{table}[h]
\centering
\caption{{\footnotesize Sample complexity upper and lower bounds for $d_1$-{\em vs}-$d_2$ {\em identity testing} of product distributions for various distance measures. First column (row) lists $d_1(P,Q) \leq \epsilon_1$ (respectively, $d_2(P,Q) > \epsilon_2$). The problem becomes computationally more difficult, and hence the sample complexity is non-decreasing, as we traverse the table down or to the right.  \ignore{For cells shaded green, the sample complexity is at most $O(|\Sigma|^{3/2}\sqrt{n}/\eps^2)$. For cells shaded yellow, the sample complexity is at most $O(|\Sigma|n\log n/\eps^2)$. For cells shaded yellow or red, the sample complexity is at least $\Omega(n/\log n)$, even for $|\Sigma|=2$. For cells shaded gray, the problem is untestable.} [\textdagger], [\textasteriskcentered] and [\ddag] are from \citet*{DBLP:conf/colt/DaskalakisP17}, \citet*{DBLP:conf/colt/CanonneDKS17} and \citet*{bhattacharyya2020efficient}, respectively. Note that for some of the cells, to get to the bound we need to follow a chain of directions.} }
\label{table:identity}
\resizebox{.9\textwidth}{!}{
\begin{tabular}{lll}
\toprule
  & $\dtv(P,Q) > \epsilon$ & $\sqrt{2}\dhel(P,Q) > \epsilon$ \\ \midrule
$P=Q$ &  {\em UB}~:~$O(\sqrt{n}/\epsilon^2)$ (for $|\Sigma|=2$) [\textdagger,\textasteriskcentered]  & {\em UB}~:~\mbox{\em Below}\\

& {\em LB}~:~$\Omega(\sqrt{n}/\epsilon^2)$  (for $|\Sigma|=2$) [\textdagger,\textasteriskcentered] & {\em LB}~:~{\em Left}\\ 

& {\em LB}~:~$\Omega( \sqrt{n\lvert \Sigma \rvert}/\epsilon^2)$ (for $|\Sigma|>2$) \cref{thm:main}&\\
\addlinespace
$\dchisqr(P,Q) \leq \epsilon^2/9$ & {\em UB}~:~{\em Right} & {\em UB}~:~$O( \sqrt{n\lvert \Sigma \rvert}/\epsilon^2)$ \cref{1S-UB} \\ 
& {\em LB}~:~{\em Above} & {\em LB}~:~{\em Left}\\
\addlinespace
\midrule[.1pt]
$\dkl(P,Q) \leq \epsilon^2/9$ & {\em UB}~:~{\em Below} & {\em UB}~:~{\em Below}\\
& {\em LB}~:~$\Omega(n / \log n)$ \cref{1S-LB} & {\em LB}~:~{\em Left} \\ 
\addlinespace
$\sqrt{2}\dhel(P,Q) \leq \epsilon/3$ & {\em UB}~:~{\em Right} &  {\em UB}~:~$O(n|\Sigma|/\epsilon^2)$ \cref{2S-UB}  \\
 & {\em LB}~:~{\em Above} & {\em LB}~:~{\em Left} \\
\addlinespace
$\dtv(P,Q) \leq \epsilon/3$ & {\em UB}~:~$O(n|\Sigma|/\epsilon^2)$ [\ddag] & {\em Not Well Defined} \\
 & {\em LB}~:~$\Omega(n/\log n)$[*] &  
\\ \bottomrule
\end{tabular}}
\end{table}
\begin{table}[h] 
\centering
\caption{{\footnotesize Sample complexity bounds for of $d_1$-{\em vs}-$d_2$ {\em closeness testing} of product distributions. As in the case of identity testing, sample complexity is non-decreasing as we traverse the table down or to the right. [\textasteriskcentered] and [\ddag] are from \citet*{DBLP:conf/colt/CanonneDKS17} and \citet*{bhattacharyya2020efficient}, respectively.} }
\label{table:closeness}
\resizebox{\textwidth}{!}
{
\begin{tabular}{lll}
\toprule
  & $\dtv(P,Q) > \epsilon$ & $\sqrt{2}\dhel(P,Q) > \epsilon$ \\ \midrule 
$P=Q~(\Sigma=2)$ & {\em UB}~:~$O(\max(\sqrt{n}/\epsilon^2,n^{3/4}/\epsilon))$\,[*] & {\em UB}~:~ $O(n^{3/4}/\eps^2)$, \cref{thm:clexhel}\\ 
& {\em LB}~:~$\Omega(\max(\sqrt{n}/\epsilon^2,n^{3/4}/\epsilon))$ [*] & {\em LB}~:~{\em Left}\\ 
\addlinespace
{\em (Any} $\Sigma)$ & {\em UB}~:~$O\!\left(\max \left\{ \sqrt{n \lvert \Sigma \rvert}/\epsilon^2, (n \lvert \Sigma \rvert)^{3/4}/\epsilon \right\}\right)$ \cref{thm:nontolDTV}& {\em UB}~:~$O((n|\Sigma|)^{3/4}/\eps^2)$, \cref{thm:clexhel} \\ 
& {\em LB}~:~{\em Above} & {\em LB}~:~{\em Above} \\
\addlinespace
\midrule[.1pt]
$\dchisqr(P,Q) \leq \epsilon^2/9$ &  {\em UB}~:~{\em Below} & {\em UB}~:~{\em Below} \\
& {\em LB}~:~ $\Omega(n/\log n)$  \cref{2S-LB} & {\em LB}~:~{\em Left} \\ 
\addlinespace

$\dkl(P,Q) \leq \epsilon^2/9$ & {\em UB}~:~{\em Below}& {\em UB}~:~{\em Below} \\ 
& {\em LB}~:~{\em Above}& {\em LB}~:~{\em Left}\\
\addlinespace
$\sqrt{2}\dhel(P,Q) \leq \epsilon/3$ & {\em UB}~:~{Right}  & {\em UB}~:~$O(n|\Sigma|/\epsilon^2)$ \cref{2S-UB}  \\ 
& {\em LB}~:~{\em Above} & {\em LB}~:~{\em Left} \\
\addlinespace
$\dtv(P,Q) \leq \epsilon/3$ & {\em UB}~:~$O(n|\Sigma|/\epsilon^2)$ [\ddag]
& {\em Not Well-defined} \\ 
& {\em LB}~:~$\Omega(n/\log n)$[*] & \\ 
\bottomrule
\end{tabular}
}
\end{table}

We informally present our main results below.  We would like to note that the only algorithmic results known regarding the complexity of testing product distributions prior to our work are: 
\begin{enumerate}[itemsep=0pt]
\item[(1)] $\Theta(\sqrt{n}/\epsilon^2)$ sample complexity bound  for the non-tolerant identity testing over the binary alphabet \citep*{DBLP:conf/colt/DaskalakisP17, DBLP:conf/colt/CanonneDKS17}, 
\item[(2)] $\Theta(\max(\sqrt{n}/\epsilon^2,n^{3/4}/\epsilon))$ sample complexity bound for non-tolerant closeness testing problem over the binary alphabet \citep*{DBLP:conf/colt/CanonneDKS17}, and
\item[(3)] $O(n|\Sigma|/\epsilon^2)$ upper bound for  $\dtv$-vs-$\dtv$ tolerant identity and closeness testing \citep*{bhattacharyya2020efficient}. 
\end{enumerate}

\subsubsection*{Identity testing: $P$ unknown and $Q$ given}
\begin{itemize}[itemsep=0pt]

\item[-]
We present a tolerant identity testing algorithm that distinguishes $\dchisqr(P,Q)\le\epsilon^2/9$ versus $\dhel(P,Q) > \epsilon$ with  $O( \sqrt{n\lvert \Sigma \rvert}/\epsilon^2)$ sample complexity. 
Since the condition for the tester rejecting $\dhel(P,Q)>\epsilon$ is stronger than $\dtv(P,Q)>\epsilon$ due to \cref{eqn:Reldistances}, we get the same bound when the second distance is $\dtv$. 
Our algorithm applies for an arbitrary $\Sigma$ and has optimal dependence on $n,\lvert \Sigma \rvert$ and $\epsilon$. 
\item[-] We present an algorithm for $\dhel$-vs-$\dhel$ identity testing with sample complexity $O(n|\Sigma|/\epsilon^2)$.  
\item[-] Our third result is a lower bound: we establish the optimality of the sample complexity of our non-tolerant identity tester w.r.t. $\dtv$in terms of $n,\lvert \Sigma \rvert$ and $\epsilon$. Such a lower bound was previously known only for $\lvert\Sigma\rvert=2$.
\item[-] Our next result is another lower bound: we show that the identity testing problem of distinguishing $\dkl(P,Q) \leq \epsilon^2/9$ versus $\dtv(P,Q) > \epsilon$ requires at least $\Omega(n/\log n)$ samples.  This shows a jump in sample complexity when we move to $\dkl$ from  $\chi^2$ distance.  Previously, \citet*{DBLP:conf/colt/CanonneDKS17} had shown the $\Omega(n/\log n)$ lower bound for $\dtv$-vs-$\dtv$ testing; we strengthen it to $\dkl$-vs-$\dtv$ testing.
\end{itemize}

\subsubsection*{Closeness testing: $P$ and $Q$ unknown}
\ignore{Next, we consider the tolerant closeness testing problem of product distributions.  For the non-tolerant case~\citet*{DBLP:conf/colt/CanonneDKS17}  showed an algorithm for closeness testing with sample complexity $O(\max(\sqrt{n}/\epsilon^2, n^{3/4}/\epsilon))$ over $\{0,1\}^n$. Our results are summarized in  below.}

\begin{itemize}[itemsep=0pt]
\item[-]   
We design an efficient algorithm that distinguishes $\sqrt{2}\dhel(P,Q)\le\epsilon/3$ versus $\sqrt{2}\dhel(P,Q)>\epsilon$ with $O(n\lvert \Sigma \rvert / \epsilon^2)$ sample complexity.  (Note that this result appears in both \cref{table:identity} and \cref{table:closeness}). Our upper bound works for distributions over arbitrary alphabet. 
\ignore{
Along the way to proving this result, we obtain a learning algorithm that learns unstructured distributions with respect to the Hellinger distance:  Given sample access to an unknown distribution $P$ over a domain of size $N$, we show a learning algorithm that takes $O(N \log (1/\delta) / \epsilon^2)$ samples and outputs a distribution $R$ which satisfy $\dhel(P,R) < \epsilon$ with probability at least $1 - \delta$.

This learning algorithm is of independent interest.  A straightforward application is to the problem of learning causal models considered in~\citet*{AcharyaBDK18}, where they showed that when the in-degree and the c-component size of a causal graph are bounded, then it is possible to learn  up to accuracy $\epsilon$ a causal model defined over the graph using $\tilde{O}(n^2/\epsilon^4)$ samples.  Their sample complexity essentially boils down to the problem of learning an unstructured distribution of constant support size, up to accuracy $\epsilon$, with error at most $1/cn$, where $c$ is some constant.  Applying our learning algorithm to this subroutine results in an algorithm that takes $\tilde{O}(n/\epsilon^2)$ samples for learning causal models.  
}
\item[-]
We complement the above upper bound with a new lower bound. We show that given sample access to two unknown distributions $P$ and $Q$, distinguishing $\dchisqr(P,Q)\leq \eps^2/9$ from $\dtv(P,Q)>\eps$ requires $\Omega(n/\log n)$ samples, even for $|\Sigma|=2$ and constant $\epsilon$. Note that this is in contrast to identity testing, where \cref{table:identity} shows that the same problem can be solved using $O( \sqrt{n\lvert \Sigma \rvert}/\eps^2)$ samples.
This also strengthens the $\dtv$-vs-$\dtv$ lower bound of \citet*{DBLP:conf/colt/CanonneDKS17}.
\end{itemize}

We also establish new upper bounds for non-tolerant closeness testing over arbitrary alphabet. Prior work considered only the binary alphabet and the extension to arbitrary alphabet is not completely straightforward. 

\ignore{
We would like to note that while many of upper and lower bounds are close to optimal, there are several gaps to be filled for completing the sample complexity map of testing product distributions which remains the main open problem that emerge from our work. 
}

\subsubsection*{Tolerant Testing for Bayes Nets}

A more general class of probability distributions, containing product distributions as a special case, is bounded-degree {\em Bayesian networks} (or Bayes nets in short). Formally, a probability distribution $P$ over $n$ variables $X_1, \dots, X_n \in \Sigma$ is said to be a {\em Bayesian network on a directed acyclic graph $G$} with $n$ nodes if\footnote{We use the notation $X_S$ to denote $\{X_i : i \in S\}$ for a set $S \subseteq [n]$.} for every $i \in [n]$, $X_i$ is conditionally independent of $X_{\text{non-descendants}(i)}$ given $X_{\text{parents}(i)}$. Equivalently, $P$ admits the factorization:
\begin{equation}\label{eqn:bnfactor}
\Pr_{X \sim P}[X=x]= \prod_{i=1}^n \Pr_{X\sim P}[X_i = x_i \mid \forall j \in {\rm parents}(i), X_j = x_j] \qquad \text{for all } x \in \Sigma^n
\end{equation}
For example, product distributions are Bayes nets on the empty graph. A {\em degree-$d$ Bayes net} is a Bayes net on a graph with in-degree bounded by $d$. 

We consider tolerant closeness testing of degree-$d$ Bayes nets on known directed acyclic graphs. 
\citet*{bhattacharyya2020efficient} designed an algorithm for tolerant $\dtv$-vs-$\dtv$ closeness testing with  $\tilde{O}(|\Sigma|^{d+1}n \eps^{-2})$ sample complexity. Our main result for Bayes nets extends this same bound to $\dhel$-vs-$\dhel$ testing, which is the hardest variant of the tolerant testing problems considered above. Moreover, our test is computationally efficient (in terms of time complexity).
Note that a computationally inefficient test readily follows from available {\em learning} algorithms for fixed-structure Bayes nets with respect to KL divergence \citep*{dasgupta1997sample, bhattacharyya2020efficient}. Indeed, the main technical component in our result is a novel efficient estimator for Hellinger distance between two distributions when given access to samples generated from them as well as their probability mass functions. This estimator may be of independent interest.

\subsection{Related Work}
The history of identity tests goes back to Pearson's chi-squared test in 1900. The traditional spirit of analyzing such tests is to consider a fixed distribution $P$  and to let the number of samples go to infinity. Work on understanding the performance of hypothesis tests with a finite number of samples mostly started only quite recently. \citet*{GoldreichR11} studied the problem of distinguishing whether an input distribution $P$ is uniform over its support or $\eps$-far from uniform in total variation distance ({in fact, they showed a {\em tolerant} tester with respect to the $\ell_2$-norm}). Paninski showed that $\Theta(\sqrt{m}/\eps^2)$ samples are necessary for uniformity testing, and gave an optimal tester when $\epsilon>m^{-1/4}$ (where $m$ is the size of the support). For the more general problem of testing identity to an arbitrary given distribution,~\citet*{BatuFRSW13} showed an upper bound of $\tilde{O}(\sqrt{m}/\eps^6)$. This was then refined by \citet*{Valiant:2014:AIP:2706700.2707449} to the tight bound of $\Theta(\sqrt{m}/\eps^2)$. \citet*{BatuFRSW13} also studied the problem of testing closeness between two input distributions and showed an upper bound of $\tilde{O}(m^{2/3} \textrm{poly}(1/\eps))$ on the sample complexity. The tight bound of $\Theta(\max(m^{2/3}\eps^{-4/3}, \sqrt{m}\eps^{-2}))$ was achieved by~\citet*{ChanDVV14}. Tolerant versions of uniformity, identity, and closeness testing with respect to the total variation distance require $\Omega(m/\log m)$ samples \citet*{ValiantV11}, which is also tight \citet*{ValiantV10}. To circumvent this lower bound, tolerant identity testing with respect to chi-squared distance was initiated by~\citet*{AcharyaDK15} and was thoroughly studied in~\citet*{Daskalakis:2018:DDS:3174304.3175479} for a number of pairs of distances.

The study of testing distributions over high-dimensional domains was initiated recently independently and concurrently in~\citet*{DaskalakisDK18,DBLP:conf/colt/CanonneDKS17,DBLP:conf/colt/DaskalakisP17}, who recognized that since testing arbitrary distributions over $\Sigma^n$ would require an exponential number of samples, it is important to make structural assumptions on the distribution. In particular, in \citet*{DaskalakisDK18}, they make the assumption that the input distributions are drawn from an Ising model. In~\citet*{DBLP:conf/colt/CanonneDKS17} and~\citet*{DBLP:conf/colt/DaskalakisP17}, the authors considered identity testing and closeness testing for distributions given by Bayes networks of bounded in-degree. \ignore{Specifically, they showed that if $P$ and $Q$ are two distributions on the same Bayes network of in-degree $d$, then they can be non-tolerantly tested for closeness using $\tilde{O}(|\Sigma|^{3(d+1)/4} n/\eps^2)$ samples.} These works also considered the special case of product distributions (equivalently, distributions over a Bayes network consisting of isolated nodes). It's shown that $\Theta(\sqrt{n}/\eps^2)$ and $\Theta(\max(\sqrt{n}/\eps^2, n^{3/4}/\eps))$ samples are necessary and sufficient for identity testing and closeness testing respectively of pairs of product distributions when $|\Sigma|=2$. The identity tester of~\citet*{DBLP:conf/colt/CanonneDKS17} 
is claimed to have certain weaker ($O(\epsilon^2)$ in $\dtv$, see Remark 8) tolerance. A reduction from testing problems for product distributions over alphabet $\Sigma$, to that for the Bayes nets of degree $\lfloor\log_2 |\Sigma|\rfloor-1$, was given in~\citet*{DBLP:conf/colt/CanonneDKS17} (Remark 55 of their paper). \citet*{DBLP:conf/colt/CanonneDKS17} also show that for product distributions, $\Omega(n/\log n)$ samples are necessary for tolerant identity and closeness testing with respect to the total variation distance. Very recently, \citet*{bhattacharyya2020efficient} designed tolerant testers for certain classes of high-dimensional distributions (including product distributions) with respect to $\dtv$.

\section{Preliminaries and Formal Statements of Results}\label{prelims}

We use $\bern(\delta)$ to denote  the Bernoulli distribution with $\mathrm{Pr}[1]=\delta$.
We define various  distance measures between distributions that we use in this paper.

\begin{definitionRe}\label{def:dists} Let $P=(p_1,p_2,\ldots,p_m)$ and $Q = (q_1,q_2,\ldots, q_m)$ be two distributions over sample space $[m]$. Then the distance measures, total variational distance, chi-squared distance, Hellinger distance, and KL distance,  respectively are defined as follows.  
$$\dtv(P,Q)=\frac12\sum_i |p_i-q_i|; \qquad
\dchisqr(P,Q)=\sum_i (p_i-q_i)^2/q_i=\sum_i p_i^2/q_i -1;$$
$$\dhelsqr(P,Q)={1\over 2}\sum_i(\sqrt{p_i}-\sqrt{q_i})^2=1-\sum_i \sqrt{p_iq_i};\qquad
\dkl(P,Q)=\sum_i p_i \ln {p_i\over q_i}$$
\end{definitionRe}

\begin{lemmaRe}\label{lem-distcompare} (folklore, see \citet*{Daskalakis:2018:DDS:3174304.3175479} for a proof) For two distributions $P$ and $Q$, the following relation holds.
\[\dhelsqr(P,Q) \leq \dtv(P,Q) \leq \sqrt{2} \dhel(P,Q) \leq \sqrt{\dkl(P,Q)} \leq \sqrt{\dchisqr(P,Q)} \]
\end{lemmaRe}

\subsection{Formal Statements of Main Results}
Here we list the formal statements of the main theorems we prove in the paper. First we state the two main upper bounds.

\begin{restatable}{theoremRe}{sUB}\mbox{\rm ($\dchisqr$-versus-$\dhel$ identity tester)}\label{1S-UB}
There is an algorithm with sample access to an unknown product distribution $P=\prodp$ and input a known product distribution $Q=\prodq$, both over the common sample space $\Sigma^n$, that decides between cases $\dchisqr(P,Q)\le \epsilon^2/9$ versus $\sqrt{2}\dhel(P,Q)>\epsilon$. The algorithm takes $O(\sqrt{n|\Sigma|}/\epsilon^2)$ samples from $P$ and runs  in time $O(n\ell+n^{3/2}\sqrt{\ell}/\epsilon^2)$. The algorithm has a success probability at least $2/3$. 
\end{restatable}
\begin{restatable}{theoremRe}{ssUB}\mbox{\rm ($\dhel$-versus-$\dhel$ closeness tester)}\label{2S-UB}
There is an algorithm with sample access to two unknown product distribution $P=\prodp$ and $Q=\prodq$, both over the common sample space $\Sigma^n$, that decides between cases $\sqrt{2}\dhel(P,Q)\le \epsilon/3$  versus $\sqrt{2}\dhel(P,Q)> \epsilon$. The algorithm takes $O(n(|\Sigma|+ \log n)/\epsilon^2)$ samples from $P$ and $Q$ and runs in time  $O(n^2(|\Sigma|+ \log n)/\epsilon^2)$. The algorithm has a success probability at least  $2/3$.
\end{restatable}

We complement the above upper bounds on sample complexity with the following lower bounds. 

\begin{restatable}{theoremRe}{thmMain}\label{thm:main}
Uniformity testing with w.r.t. $\dtv$ distance for product distributions over $[\ell]^n$ needs $\Omega(\sqrt{n\ell}/\epsilon^{2})$ samples.
\end{restatable}

\begin{restatable}{theoremRe}{sLB}\mbox{\rm ($\dkl$-versus-$\dtv$ identity testing lower bound)}\label{1S-LB}
There exists a constant $0<\epsilon<1$ and three product distributions $F^{yes}, F^{no}$ and $F$, each over the sample space $\{0,1\}^n$ such that $\dkl(F^{yes},F)\le \epsilon^2/9$, whereas $\dtv(F^{no},F)> \epsilon$, and given only sample accesses to $F^{yes}, F^{no}$, and complete knowledge about $F$, distinguishing $F^{yes}$ versus $F^{no}$ with probability $> 2/3$, requires $\Omega(n/\log n)$ samples.
\end{restatable}

\begin{restatable}{theoremRe}{ssLB}\mbox{\rm ($\dchisqr$-versus-$\dtv$ closeness testing lower bound)}\label{2S-LB}
There exists a constant $0<\epsilon<1$ and three product distributions $F^{yes}, F^{no}$ and $F$, each over the sample space $\{0,1\}^n$ such that $\dchisqr(F^{yes},F)\le \epsilon^2/9$, whereas $\dtv(F^{no},F)> \epsilon$, and given only sample accesses to $F^{yes}, F^{no}$ and $F$, distinguishing $F^{yes}$ versus $F^{no}$ with probability $> 2/3$, requires $\Omega(n/\log n)$ samples.
\end{restatable}

Earlier work has designed non-tolerant closeness tester for product distribution over a binary alphabet. Here we extend it to arbitrary alphabets. 
\begin{restatable}{theoremRe}{clexhel}\mbox{\rm (Exact-versus-$\dhel$ closeness tester)}\label{thm:clexhel}
There is an algorithm with sample access to two unknown product distribution $P=\prodp$ and $Q=\prodq$, both over the common sample space $\Sigma^n$, that decides between cases 
$P=Q$ versus $\sqrt{2}\dhel(P,Q)>\epsilon$. The algorithm takes $m=O((n|\Sigma|)^{3/4}/\epsilon^2)$ samples from $P$ and $Q$ and runs in time $O(mn)$. The algorithm has a success probability at least $2/3$.
\end{restatable}  

\begin{restatable}{theoremRe}{nontolDTV}\mbox{\rm (Exact-versus-$\dtv$ closeness tester)}
\label{thm:nontolDTV}
There is an algorithm with sample access to two unknown product distribution $P=\prodp$ and $Q=\prodq$, both over the common sample space $\Sigma^n$, that decides between cases 
 $P=Q$ versus $\dtv(P,Q)>\epsilon$. The algorithm takes $m=O(\max\{\sqrt{n|\Sigma|}/\epsilon^2,(n|\Sigma|)^{3/4}/\epsilon\})$ samples and runs in time $O(mn)$. The algorithm has a  success probability at least $2/3$.
\end{restatable}

Finally, we state our result for $\dhel$-vs-$\dhel$ closeness testing of fixed-structure Bayes nets.
\begin{restatable}{theoremRe}{bnTolH}\mbox{\rm ($\dhel$-vs-$\dhel$ closeness tester for Bayes nets)}
\label{thm:bnlTOlH}
Given samples from two unknown Bayesian networks $P$ and $Q$ over $\Sigma^n$ on potentially different but known pair of graphs of indegree at most $d$, we can distinguish the cases $\dhel(P,Q)\le \epsilon/2$ versus $\dhel(P,Q)> \epsilon$ with 2/3 probability using $m=O(|\Sigma|^{d+1}n\log(|\Sigma|^{d+1} n)\epsilon^{-2})$ samples and $O(|\Sigma|^{d+1}mn+n\epsilon^{-4})$ time.
\end{restatable}

\section{Efficient Tolerant Testers}\label{sec-upper}

\subsection{$\dchisqr$-vs-$\dhel$ Tolerant Identity Tester}\label{sec:chisq-1S}

In this section, we generalize the testers of~\citet*{DBLP:conf/colt/DaskalakisP17,DBLP:conf/colt/CanonneDKS17} that distinguishes $P=Q$ (`yes class') versus $\dtv(P,Q)\ge \epsilon$ (`no class') using $O(\sqrt{n}/\epsilon^2)$ samples, where $P$ and $Q$ are product distributions over $\{0,1\}^n$. Our first contribution is to generalize their tester in the following three ways. Firstly, our `no class' is defined as $\sqrt{2}\dhel(P,Q)\ge \epsilon$, which is more general than $\dtv(P,Q)\ge \epsilon$. Secondly, our tester works for any general alphabet size $|\Sigma| \ge 2$. Finally, we give a $\dchisqr$ tolerant tester i.e. our `yes class' is defined as $\dchisqr(P,Q)< \epsilon^2/9$.

Our tester relies on certain factorizations of $\dchisqr$ and $\dhelsqr$ for product distributions.  We shall proceed to discuss those relations.

\begin{restatable}{lemmaRe}{chisqmult}\label{lem-chisq-mult}(folklore, see \citet*{AcharyaDK15} for a proof)
Let $P=\prodp$ and $Q=\prodq$
be two distributions, both over the common sample space $\Sigma^n$.
Then $\dchisqr(P,Q)=\prod_{i=1}^n (1+\dchisqr(P_i,Q_i))-1$. In particular, $\dchisqr(P,Q)\ge\sum_{i} \dchisqr(P_i,Q_i)$
\end{restatable}

\begin{restatable}{factRe}{helequality}\label{lem-helloc-equality}(folklore)
Let $P=\prod_{i=1}^n P_i$ and $Q=\prod_{i=1}^n Q_i $ be two distributions over $\Sigma^n$. It holds that $1-\dhelsqr(P,Q)= \Pi_{i=1}^n (1-\dhelsqr(P_i,Q_i))$. In particular, $\dhel^2(P,Q) \le \sum_i \dhel^2(P_i,Q_i)$.
\end{restatable}

We get the following useful corollary from \cref{lem-distcompare} and \cref{lem-helloc-equality}.
\begin{corollaryRe}\label{cor-noclass}
Let $P=\prodp$ and $Q=\prodq$ be two distributions, both over the common sample space $\Sigma^n$.
Then $\dhelsqr(P,Q)\le\sum_{i} \dchisqr(P_i,Q_i)/2$.
\end{corollaryRe}

To avoid low probabilities in the denominator of the test statistic, we need to ensure that for each distribution $Q_i$, each element in the sample space $\Sigma$ gets at least a sufficiently large probability $\Omega(\epsilon^2/|\Sigma| n)$. We do this by {\em slightly randomizing} $Q$ to get a new distribution $S$. The randomization process to get $S$ from $Q$ is given below. This is similar to the reduction given in~\citet*{DBLP:conf/colt/DaskalakisP17} and~\citet*{DBLP:conf/colt/CanonneDKS17} for the case $\Sigma = \{0,1\}$ and for the case when the `no' class is defined with respect to $\dtv$. Let $\bern(\delta)^n$ be the product distribution of $n$ copies of $\bern(\delta)$.

\begin{restatable}{lemmaRe}{modification}\label{lem-modification}
For a product distribution $P = \prod_{i=1}^n P_i$, where the $P_i's$ are over a sample space $\Sigma$, and $0<\delta < 1$, let $P^{\delta}$ be the distribution over $\Sigma^n$ defined by the following sampling process. In order to produce a sample $(X_1,X_2,\ldots,X_n)$ of $P^{\delta}$, 
\begin{enumerate}[itemsep=0mm]
\item[-] Sample $(r_1,r_2,\dots,r_n)\sim \bern(\delta)^n$ and sample $(Y_1,Y_2,\dots,Y_n)\sim P$
\item[-] For every $i$, if $r_i=1$, $X_i\leftarrow$ uniform sample from $\Sigma$, if $r_i=0$, $X_i \leftarrow Y_i$. 
\end{enumerate} 
Then, the following is true. 
\begin{enumerate}[itemsep=0mm] 
\item[-] $P^{\delta}$ is a product distribution $\prod_{i} P_i^{\delta}$ and each sample from $P^{\delta}$ can be simulated by 1 sample from $P$.
\item[-] For every $i: 1\leq i \leq n$ and $j\in \Sigma$, $P_i^{\delta}(j)\ge \delta/|\Sigma|$. 
\item[-] $\dhelsqr(P,P^{\delta}) \leq 2n\delta$
\end{enumerate}
\end{restatable}
\begin{proof}
The first part is obvious from the sampling process. For the second part, $P_i^{\delta}(j)= (1-\delta)P_i(j)+\delta/|\Sigma| \ge \delta/|\Sigma|$ for every $i,j$. 

The proof of the third part can be obtained by generalizing the proof of \citet*{Daskalakis:2018:DDS:3174304.3175479}. Consider the $i$-th component of $P$ and $P^\delta$, denoted $P_i$ and $Q_i$ respectively for convenience. Let $E_i$ be the event that $r_i=0$. Also note that conditioned on the event $E_i$, for any item $j\in\Sigma,$ the probability values satisfy $Q_i(j\mid E_i)=P_i(j)$.
\begin{align*}
\dhelsqr(P_i,Q_i) &= \sum_{j\in \Sigma} \left( \sqrt{Q_i(j)} -\sqrt{P_i(j)} \right)^2\\
&=  \sum_{j\in \Sigma} \left( \sqrt{Q_i(j \mid E_i)\mathrm{Pr}(E_i)+Q_i(j \mid \bar{E}_i)\mathrm{Pr}(\bar{E}_i)} -\sqrt{P_i(j)} \right)^2\\
&=  \sum_{j} \left( \sqrt{P_i(j)\mathrm{Pr}(E_i)} -\sqrt{P_i(j)} \right)^2+\sum_j Q_i(j \mid \bar{E}_i)\mathrm{Pr}(\bar{E}_i)&&\tag{Using $(\sqrt{a+b}-\sqrt{c+d})^2\le (\sqrt{a}-\sqrt{c})^2+(\sqrt{b}-\sqrt{d})^2$ for non-negative $a,b,c,d$}\\
&=(1-\sqrt{\mathrm{Pr}(E_i)})^2+\mathrm{Pr}(\bar{E}_i)\\
&=(1-\sqrt{1-\mathrm{Pr}(\bar{E}_i)})^2+\mathrm{Pr}(\bar{E}_i)\\
&\le 2\mathrm{Pr}(\bar{E}_i)=2\delta&&\tag{Using $(1-\sqrt{1-x})^2\le x$ for $0\le x \le1$}\\
\end{align*}

Then \cref{lem-helloc-equality} gives us $\dhelsqr(P,P^\delta) \le 2n\delta$ due to sub-additivity.
\ignore{
For convenience, we denote $P^{\delta}$ by $R$. Define the event $E$ to be $a_i=0$ for every $i$. Then $\mathrm{Pr}(E)=(1-\delta)^n \ge (1-n\delta)$ and $\mathrm{Pr}[\bar{E}]\le n\delta$. Also note that conditioned on the event $E$, for any item $i\in\Sigma^n,$ the probability values satisfy $R(i|E)=P(i)$.

\begin{align*}
\dhelsqr(P,R)&=\sum_i (\sqrt{R(i)}-\sqrt{P(i)})^2\\
&=\sum_i (\sqrt{R(i|E)\mathrm{Pr}(E)+R(i|\bar{E})\mathrm{Pr}(\bar{E})}-\sqrt{P(i)})^2\\
&=\sum_i (\sqrt{P(i)\mathrm{Pr}(E)+R(i|\bar{E})\mathrm{Pr}(\bar{E})}-\sqrt{P(i)})^2\\
&\le\sum_i [(\sqrt{P(i)\mathrm{Pr}(E)}-\sqrt{P(i)})^2+(\sqrt{R(i|\bar{E})\mathrm{Pr}(\bar{E})})^2]
&&\tag{Using $(\sqrt{a+b}-\sqrt{c+d})^2\le (\sqrt{a}-\sqrt{c})^2+(\sqrt{b}-\sqrt{d})^2$ for non-negative $a,b,c,d$}\\
&=(1-\sqrt{\mathrm{Pr}(E)})^2\sum_i P(i)+\mathrm{Pr}(\bar{E})\sum_i R(i|\bar{E})\\
&=(1-\sqrt{1-\mathrm{Pr}(\bar{E})})^2+\mathrm{Pr}(\bar{E})\\
&\le 2\mathrm{Pr}(\bar{E}) &&\tag{Using $(1-\sqrt{1-x})^2\le x$ for $0\le x \le1$}\\
&\le 2n\delta
\end{align*}
}
\end{proof}

\begin{restatable}{lemmaRe}{RSfromPQ}\label{lem-RSfromPQ}
Let $P=\prod_{i=1}^{n}P_{i}$ and $Q=\prod_{i=1}^n Q_i$ be two distributions, both over the common sample space $\Sigma^n$. Let $\ell=|\Sigma|$, $R=P^{\delta}$, $S=Q^{\delta}$  with $\delta=\epsilon^2/50n$.
 $R=\prod_{i=1}^{n}R_i$ and $S=\prod_{i=1}^n S_i$, where
$R_i=\langle r_{i1},r_{i2},\dots, r_{i\ell}\rangle$ and $S_i=\langle s_{i1},s_{i2},\dots, s_{i\ell}\rangle$ for
every $i$. Then
\begin{itemize}[itemsep=0mm]
\item[(1)] If $\dchisqr(P,Q) \le \epsilon^2/9$ then $\sum_{i,j} {(r_{ij}-s_{ij})^2\over s_{ij}} < 0.12\epsilon^2$.
\item[(2)] If $\sqrt{2}\dhel(P,Q)\ge \epsilon$ then $\sum_{i,j} {(r_{ij}-s_{ij})^2\over s_{ij}} > 0.18\epsilon^2$.
\end{itemize}
\end{restatable}
\begin{proof}
({\it Proof of (1)}) We have that $r_{ij}=(1-\delta)p_{ij}+\delta/\ell$ and $s_{ij}=(1-\delta)q_{ij}+\delta/\ell$. Then,
$
\sum_{i,j} {(r_{ij}-s_{ij})^2\over s_{ij}} = \sum_{i,j} {(1-\delta)^2(p_{ij}-q_{ij})^2\over (1-\delta)q_{ij}+\delta/\ell}
\le \sum_{i,j} {(1-\delta)^2(p_{ij}-q_{ij})^2\over (1-\delta)q_{ij}}
= (1-\delta) \sum_{i,j} {(p_{ij}-q_{ij})^2\over q_{ij}}
< \sum_{i}\sum_{j} \ab {(p_{ij}-q_{ij})^2\over q_{ij}}
= \sum_{i}\dchisqr(P_i,Q_i)
\le \dchisqr(P,Q) 
< 0.12\epsilon^2
$
. The second last step is due to \cref{lem-chisq-mult}.

\noindent({\it Proof of (2)}). From \cref{lem-modification}, for $\delta=\epsilon^2/50n$,  it  follows that $\dhelsqr(P,R)\le \epsilon^2/25$ and $\dhelsqr(Q,S)\le \epsilon^2/25$. By triangle inequality we get
$\dhel(P,Q)\le \dhel(R,S)+\dhel(P,R)+\dhel(Q,S)$. It follows that if $\sqrt{2}\dhel(P,Q)\ge\epsilon$ then $\dhel(R,S)\ge \epsilon(1/\sqrt{2}-2/5)$.
Then \cref{cor-noclass} gives $\sum_{i,j} {(r_{ij}-s_{ij})^2\over s_{ij}}=\sum_{i} \dchisqr(R_i,S_i)) \ge 2\dhelsqr(R,S) > 0.18\epsilon^2$.
\end{proof}

At this point it remains to test $\sum_{i,j} {(r_{ij}-s_{ij})^2\over s_{ij}} > 0.18\epsilon^2/10$ versus $< 0.12\epsilon^2/10$, which we perform using the tester of~\citet*{AcharyaDK15}. We state their tester with necessary modifications and prove it in the Appendix.  

\begin{restatable}{theoremRe}{adk}\mbox{\rm (Modified from~\citet*{AcharyaDK15})}\label{thm-adk} Let $m$ be an integer and $0 < \epsilon < 1$ 
be an error parameter. Let $r_1, r_2,\ldots, r_K$ be $K$ non-negative real numbers. Let  $s_1,s_2,\ldots,s_K$ be non-negative real numbers such that $s_i \geq \epsilon^2/50K$.  
For $1\leq i \leq K$, let $N_i \sim {\rm Poi}(mr_i)$ be independent samples from ${\rm Poi}(mr_i)$. Then there exists a test statistic $T,$ computable in time $O(K)$ from inputs $N_i$s and $s_i$s, with the following guarantees. 
\begin{itemize}
\item[-] $\mathrm{E}[T]=m\sum_i {(r_i-s_i)^2\over s_i}$
\item[-] $\mathrm{Var}[T]\le 2K + 7\sqrt{K}\e[T]+4K^{1/4}(\e[T])^{3/2}$, for a constant $c$ and $m\ge c\sqrt{K}/\epsilon^2$.  
\end{itemize}
\end{restatable}

\paragraph{Remark} The test $T$ of~\citet*{AcharyaDK15} is given by $T=\Sigma_{i=1}^n {(N_i-ms_i)^2 - N_i\over ms_i}$.    
Their paper gives the upper bound 
$\mathrm{Var}[T]\le 4n+9\sqrt{n}\mathrm{E}[T]+{2\over 5}n^{1\over 4}\mathrm{E}[T]^{3/2}$ under the assumption $s_i \ge \epsilon/50n$ for every $i$. In our application, $\ell$ is the alphabet size and we will need the bound to depend on $\ell$. In addition, we also need the bounds to work when $s_i \geq \epsilon^2/50n\ell$. Both these can be achieved by modifying 
their proof.

It remains to sample numbers $N_{ij}\sim \poi(mr_{ij})$ independently for every $i,j$. We do this via poissonization followed by sampling from each coordinate of the product distribution $R$ independently. We present \cref{Algo-chisqtest} with its correctness.

\begin{algorithm2e}
\caption{Given samples from an unknown distribution $R=\prod_{i=1}^nR_i$ and a known distribution $S=\prod_{i=1}^nS_i$ over $\Sigma^n$, decide $\dchisqr(R,S)\le \epsilon^2/9$ (`yes') versus $\sqrt{2}\dhel(R,S)> \epsilon$ (`no'). Let $\ell=|\Sigma|$, $R_i=\langle r_{i1},r_{i2},\dots,r_{i\ell}\rangle$, $S_i=\langle s_{i1},s_{i2},\dots,s_{i\ell}\rangle$ with $s_{ij}\ge \epsilon^2/50n\ell$ for every $j$, for every $i$}
 \label{Algo-chisqtest}
\For{$i=1$ to $n$}{
Sample $N_i\sim \poi(m)$ independently\;
}
$N=\max_i N_i$\;
$X\leftarrow$ Take $N$ samples from $R$\;
\For{$i=1$ to $n$}{
$X_i\leftarrow$ Sequence of symbols in the $i$-th coordinate of first $N_i$ samples of $X$\;
$\langle N_{i1},N_{i1},\dots,N_{i\ell} \rangle\leftarrow $ histogram of symbols in $X_i$\;
}
Compute statistic $T$ of \cref{thm-adk} using $N_{ij}$ and $s_{ij}$ values for every $i,j$\;
\eIf{$T \le 0.15 m\epsilon^2$}{output `yes.'\;}{output `no.'}
\end{algorithm2e}

\begin{proof} ({\em of \cref{1S-UB}})
Let $\ell=|\Sigma|$.
First, we transform the distributions $P$ and $Q$ into the distributions $R$ and $S$ respectively according to
the modification process mentioned in \cref{lem-RSfromPQ}. This gives:
\begin{enumerate}[itemsep=0pt]
\item[-] Each sample from $R$ can be simulated by 1 sample from $P$.
\item[-] $R$ and $S$ are product distributions, $R=\prod_{i=1}^nR_i$ and $S=\prod_{i=1}^nS_i$, where
$R_i=\langle r_{i1},r_{i2},\dots, r_{i\ell}\rangle$ and $S_i=\langle s_{i1},s_{i2},\dots, s_{i\ell}\rangle$ for
every $i$.
\item[-] For every $i,j$, $s_{ij}\ge \epsilon^2/50n\ell$.
\item[-] If $\dchisqr(P,Q) \le \epsilon^2/9$ then $\sum_{i,j} {(r_{ij}-s_{ij})^2\over s_{ij}} < 0.12\epsilon^2$.
\item[-] If $\sqrt{2}\dhel(P,Q)> \epsilon$ then $\sum_{i,j} {(r_{ij}-s_{ij})^2\over s_{ij}} > 0.18\epsilon^2$.
\end{enumerate}
Henceforth, we focus on distinguishing $\sum_{i,j} {(r_{ij}-s_{ij})^2\over s_{ij}}<0.12\epsilon^2$ versus $> 0.18\epsilon^2$, under the assumption $s_{ij}\ge \epsilon^2/50n\ell$ for every $i,j$, by sampling from $R$. We use the tester $T$ of~\citet*{AcharyaDK15} stated in \cref{thm-adk} with $K=n\ell$, for this.
Firstly, note that in \cref{Algo-chisqtest}, the samples $S_i$ is a set of $N_i\sim \poi(m)$ samples from $R_i$, independently for every $i$'s. This is because the set of samples are taken from the product distribution $R=R_1\times R_2 \times \dots \times R_n$ and the $N_i$ values are independent for different $i$'s. Due to Poissonization it follows $N_{ij}\sim \poi(r_{ij})$ independently for every $i,j$. The tester $T$ requires $m\ge c\sqrt{n\ell}/\epsilon^2$, for some constant $c$ and satisfies $\mathrm{E}[T]=m\sum_{i,j} {(r_{ij}-s_{ij})^2\over s_{ij}}$, $\mathrm{Var}[T] \le 2n\ell + 7\sqrt{n\ell}\e[T]+4(n\ell)^{1/4}(\e[T])^{3/2}$.

If $\sum_{i,j} {(r_{ij}-s_{ij})^2\over s_{ij}}<0.12\epsilon^2$ then we get $\mathrm{E}[T]\le 0.12m\epsilon^2$ and
$\mathrm{Var}[T]\le \left({2 \over c^2}+{0.84 \over c}+{4(0.12)^3\over\sqrt{c}}\right)m^2\epsilon^4$, using $m\ge c\sqrt{nl}/\epsilon^2$ and the upper bound for $\mathrm{E}[T]$. By Chebyshev's inequality $T < 0.15m\epsilon^2$ with probability at least 4/5, where $c=\Omega(1)$ is an appropriate constant.

If $\sum_{i,j} {(r_{ij}-s_{ij})^2\over s_{ij}}> 0.18\epsilon^2$ then we get $\mathrm{E}[T]> 0.18m\epsilon^2 \ge 0.18c\sqrt{nl}$ and
$\mathrm{Var}[T]\le ({2\over (0.18c)^2}+{7 \over 0.18c}+{4 \over \sqrt{0.18c}})\mathrm{E}^2[T]$. By Chebyshev's inequality $T > 0.15m\epsilon^2$ with probability at least 4/5, for an appropriate constant $c=\Omega(1)$.

Hence, for some constant $c'$, $m\ge c'\sqrt{n\ell}/\epsilon^2$ suffices for the tester $T$ to distinguish the above two cases. It also follows from the concentration of the Poisson distribution that the number of samples required is $\max_i N_i\le 2m$, except for probability at most $n\cdot \exp(-m)<1/10$ (using union bound).

The histograms can be computed by a single pass over the $n$-dimensional sample set $S$. The statistic $T$ can be computed in time 
$O(n\ell)$. So the time complexity is $O(n\ell+n^{3/2}\sqrt{\ell}/\epsilon^2)$.
\end{proof}

\subsection{$\dhel$-vs-$\dhel$ Tolerant Closeness Tester}\label{sec:heltol-2S}
In this section, we give a tester for distinguishing $\dhel(P,Q)\le\epsilon$ versus $\dhel(P,Q)>3\epsilon$ for two
unknown product distributions $P$ and $Q$ over support $\Sigma^n$. To get \cref{2S-UB}, we rescale $\epsilon$ down to $\epsilon/\sqrt{2}$. We take a testing-by-learning approach: we first learn $P$ and $Q$ in Hellinger distance $\epsilon/2$ using the following known result. Then the Hellinger distance between the learnt distributions can be computed exactly.

\begin{theoremRe} \citep*{AcharyaDK15}\label{thm:HellLearn}
Given samples from an unknown product distribution $D$ over $\Sigma^n$, $\hat{D}$, the product of component-wise empirical distributions on $m$ samples satisfy $\dhel(D,\hat{D}) \le \epsilon$ with $9/10$ probability if $m \ge \Theta(n|\Sigma|/\epsilon^2)$.
\end{theoremRe}

\ignore{
\begin{proof}
Using \cref{def:dists} for $\dhelsqr(\cdot,\cdot)$ distance:
\begin{align*}
1-\dhelsqr(P,Q)
&=\sum_{i_1\in\Sigma} \sqrt{P_1(i_1)Q_1(i_1)}\cdots \sum_{i_n\in\Sigma} \sqrt{P_n(i_n)Q_n(i_n)}
=\prod_{i=1}^n (1-\dhelsqr(P_i,Q_i))
\end{align*}
\end{proof}
}
\begin{proof} ({\em of \cref{2S-UB}})
We first learn $P$ and $Q$ as $\hat{P}$ and $\hat{Q}$ using \cref{thm:HellLearn} such that $\dhel(P,\hat{P})\le \epsilon/2$ and $\dhel(Q,\hat{Q})\le \epsilon/2$, together with 4/5 probability. Conditioned on this. we compute $\dhel(\hat{P},\hat{Q})$ exactly using \cref{lem-helloc-equality}. 

Due to triangle inequality, $\dhel(\hat{P},\hat{Q})\le 2\epsilon$ or not would decide $\dhel(P,Q)\le \epsilon$ or $>3\epsilon$.
\end{proof}

\section{Lower Bounds}\label{sec-lowerbounds}
In this section, we give lower bounds for tolerant testing of product distributions. Our lower bounds use a reduction from testing the class of unstructured distributions over $n$ items to testing the class of product distributions over $\{0,1\}^n$, given by~\citet*{DBLP:conf/colt/CanonneDKS17} (Section 4.5 of their paper). However, in order to apply this reduction, we need to establish certain new bounds relating the distances in the unstructured setting to the setting of product distribution. 
We first define how to construct a product distribution from the corresponding unstructured distribution. In particular, for a $\delta<1$, this construction produces a product distribution $F_{\delta}(P)$ over $\{0,1\}^n$ from a given distribution $P$ over $n$ symbols. 

\begin{definitionRe}\label{def-fdeltap}{\mbox{\rm (Construction of $F_\delta(P)$)}}
Let $P$ be a distribution over a sample space of $n$ items and $0<\delta\le 1$ be a constant. Let $S$ be a random set of $\poi(\delta)$ samples from $P$. For every item $i\in [n]$, let $x_i$ be the indicator variable such that $x_i=1$ iff $i$ appears in $S$. Let $F_{\delta}(P)$ be the joint distribution of $\langle x_1,x_2,\dots,x_n\rangle$ over the sample space $\{0,1\}^n$.
\end{definitionRe}

The following property can be observed using the property of Poissonization.

\begin{factRe}\label{fact-fdeltap}
Let $P$ be a distribution over a sample space of $n$ items with probability vector $\langle p_1,p_2,\dots,p_n\rangle$ and $0<\delta\le 1$ be a constant. Then $F_{\delta}(P)$ is a product distribution such that $F_{\delta}(P)=\prod_{i=1}^{n}F_{\delta}(P_i)$ where  $F_{\delta}(P_i)\sim \bern(1-e^{-\delta p_i})$.
\end{factRe}

We use the following crucial lemma.
\begin{restatable}{lemmaRe}{relFdelP}\label{lem-dtv-pfdeltap}
For any $0<\delta\le 1$ and distributions $P,Q$, $\dtv(F_{\delta}(P),F_{\delta}(Q))\ge \delta e^{-\delta} \dtv(P,Q)$, with equality holding iff $P=Q$.
\end{restatable}
\begin{proof}
Let $P=\langle p_1,\dots,p_i,\dots,p_n \rangle$ and $Q=\langle q_1,\dots,q_i,\dots,q_n \rangle$ be the probability values of $P$ and $Q$.
\begin{align*}
\dtv(F_{\delta}(P),F_{\delta}(Q))&=\sum_{x \in \{0,1\}^n} |F_{\delta}(P)(x)-F_{\delta}(Q)(x)|\\ 
&\ge \sum_{i=1}^n |F_{\delta}(P)(e_i)-F_{\delta}(Q)(e_i)| &&\tag{unit vector $e_i$ has $i$-th value 1}\\
&= \sum_{i=1}^n |(1-e^{-\delta p_i})\Pi_{j\neq i} e^{-\delta p_j}-(1-e^{-\delta q_i})\Pi_{j\neq i} e^{-\delta q_j}|\\
&= \sum_{i=1}^n e^{-\delta} |e^{\delta p_i}-e^{\delta q_i}| &&\tag{Since $\Pi_{j} e^{-\delta p_j}=\Pi_{j} e^{-\delta q_j}=e^{-\delta}$}\\
&= e^{-\delta} \sum_{i=1}^n |\delta(p_i-q_i)+\delta^2(p_i^2-q_i^2)/2!+\dots+\delta^j(p_i^j-q_i^j)/j!+\dots|.
\end{align*}
We analyze the expression under modulus under two cases: 1) if $p_i > q_i$, it is more than $\delta(p_i-q_i)$, 2) if $p_i < q_i$, it is more than $\delta(q_i-p_i)$.
\begin{align*}
\dtv(F_{\delta}(P),F_{\delta}(Q))&\ge e^{-\delta}\sum_{i=1}^n |\delta(p_i-q_i)|\\
&=\delta e^{-\delta} \dtv(P,Q).
\end{align*} 
\end{proof}

\ignore{
\begin{proof}
Let $P=\langle p_1,\dots,p_i,\dots,p_n \rangle$ and $Q=\langle q_1,\dots,q_i,\dots,q_n \rangle$ be the probability values of $P$ and $Q$.
\begin{align*}
\dtv(F_{\delta}(P),F_{\delta}(Q))&=\sum_{x \in \{0,1\}^n} |F_{\delta}(P)(x)-F_{\delta}(Q)(x)|\\ 
&\ge \sum_{i=1}^n |\dtv(F_{\delta}(P)(e_i)-F_{\delta}(Q))(e_i)| &&\tag{unit vector $e_i$ has $i$-th value 1}\\
&= \sum_{i=1}^n |(1-e^{-\delta p_i})\Pi_{j\neq i} e^{-\delta p_j}-(1-e^{-\delta q_i})\Pi_{j\neq i} e^{-\delta q_j}|\\
&= \sum_{i=1}^n e^{-\delta} |e^{\delta p_i}-e^{\delta q_i}| &&\tag{Since $\Pi_{j} e^{-\delta p_j}=\Pi_{j} e^{-\delta q_j}=e^{-\delta}$}\\
&= e^{-\delta} \sum_{i=1}^n |\delta(p_i-q_i)+\delta^2(p_i^2-q_i^2)/2!+\dots+\delta^j(p_i^j-q_i^j)/j!+\dots| 
\end{align*}
We analyze the expression under modulus under two cases: 1) if $p_i > q_i$, it is $>\delta(p_i-q_i)$,
2) if $p_i < q_i$, it is $>\delta(q_i-p_i)$.
\begin{align*}
\dtv(F_{\delta}(P),F_{\delta}(Q))&\ge e^{-\delta}\sum_{i=1}^n |\delta(p_i-q_i)|\\
&=\delta e^{-\delta} \dtv(P,Q)
\end{align*} 
\end{proof}
}

\subsection{Hardness of $\dchisqr$-vs-$\dtv$ Tolerant Closeness Testing}
Here we show that for two unknown product distribution $P,Q$ over $\{0,1\}^n$, distinguishing $\dchisqr(P,Q)\le \epsilon^2/9$ versus $\dtv(P,Q)> \epsilon$, for a constant $\epsilon$, can not be decided in general with a truly sublinear sample complexity. We use a reduction to the following difficult problem, for hardness of $\chi^2$-tolerance for closeness testing of unstructured distributions over $n$ items, given in~\citet*{Daskalakis:2018:DDS:3174304.3175479}. We restate the theorem with changes in the constants. 
\begin{theoremRe}\label{thm-unstr-chisq2sample}
There exists a constant $0<\epsilon<1$ and three distributions $P^{yes}, P^{no}$ and $Q$, each over the sample space $[n]$ such that: (1) $\dchisqr(P^{yes},Q)\le \epsilon^2/216$, whereas $\dtv(P^{no},Q)\ge \epsilon$ and (2) given only sample accesses to one of $P^{yes}$ or $P^{no}$, and $Q$, distinguishing $P^{yes}$ versus $P^{no}$ with probability $>4/5$, requires $\Omega(n/\log n)$ samples.
\end{theoremRe}

We use the following important property about the $\chi^2$-distance between the reduced distributions. 

\begin{restatable}{lemmaRe}{chisqFdelta}\label{lem-chisqFdeltaPQ}
$\dchisqr(F_{\delta}(P),F_{\delta}(Q))\le \exp(4 \delta \cdot \chi^2(P,Q))-1$, for any $0<\delta\le 1$.
\end{restatable}
\begin{proof}
From \cref{fact-fdeltap}, both $F_{\delta}(P)$ and $F_{\delta}(Q)$ are product distributions, the 
distribution of the $i$-th component being $F_{\delta}(P_i)$ and $F_{\delta}(Q_i)$ respectively. Let $P=\langle p_1,p_2,\dots,p_n\rangle$ and $Q=\langle q_1,q_2,\dots,q_n\rangle$. Then $F_{\delta}(P_i)\sim \bern(1-e^{-\delta p_i})$ and $F_{\delta}(Q_i)\sim \bern(1-e^{-\delta q_i})$. 
\begin{align*}
\dchisqr(F_{\delta}(P),F_{\delta}(Q))&=\prod_i (1+\dchisqr(F_{\delta}(P_i),F_{\delta}(Q_i)))-1
&&\tag{From \cref{lem-chisq-mult}}\\
& \le \prod_i \exp(\dchisqr(F_{\delta}(P_i),F_{\delta}(Q_i))) -1 &&\tag{Since $e^x\ge (1+x)$ for $x\ge 0$}\\
&= \exp(\sum_i \dchisqr(F_{\delta}(P_i),F_{\delta}(Q_i)))-1\\
&= \exp\left(\sum_i (e^{-\delta p_i}-e^{-\delta q_i})^2\left({1\over e^{-\delta q_i}}+{1\over 1-e^{-\delta q_i}}\right)\right) -1
&& \tag{Since $F_{\delta}(P^{yes}_i)\sim \bern(1-e^{-\delta p_i})$ and $F_{\delta}(Q_i)\sim \bern(1-e^{-\delta q_i})$}\\
&= \exp(\sum_i (e^{-\delta p_i}-e^{-\delta q_i})^2/e^{-\delta q_i}(1-e^{-\delta q_i})) -1\\
&= \exp(\sum_i (e^{\delta (q_i-p_i)}-1)^2/(e^{\delta q_i}-1)) -1\\
&= \exp(\sum_i (e^{\delta (q_i-p_i)}-1)^2/(e^{\delta q_i}-1)) -1\\
&\le \exp(\sum_i (2\delta (p_i-q_i))^2/\delta q_i)) -1
&&\tag{Since $(e^x-1)\ge x$ and $(|e^{x}-1|\le 2|x|$ for $0<|x|<1$}\\
&= \exp(4\delta \chi^2(P,Q))-1.\\
\end{align*}
\end{proof}

\ignore{
\begin{proof}
From \cref{fact-fdeltap}, both $F_{\delta}(P)$ and $F_{\delta}(Q)$ are product distributions, the 
distribution of the $i$-th component being $F_{\delta}(P_i)$ and $F_{\delta}(Q_i)$ respectively. Let $P=\langle p_1,p_2,\dots,p_n\rangle$ and $Q=\langle q_1,q_2,\dots,q_n\rangle$. Then $F_{\delta}(P_i)\sim \bern(1-e^{-\delta p_i})$ and $F_{\delta}(Q_i)\sim \bern(1-e^{-\delta q_i})$. 
\begin{align*}
\dchisqr(F_{\delta}(P),F_{\delta}(Q))&=\prod_i (1+\dchisqr(F_{\delta}(P_i),F_{\delta}(Q_i)))-1
&&\tag{From \cref{lem-chisq-mult}}\\
& \le \prod_i \exp(\dchisqr(F_{\delta}(P_i),F_{\delta}(Q_i))) -1 &&\tag{Since $e^x\ge (1+x)$ for $x\ge 0$}\\
&= \exp(\sum_i \dchisqr(F_{\delta}(P_i),F_{\delta}(Q_i)))-1\\
&= \exp\left(\sum_i (e^{-\delta p_i}-e^{-\delta q_i})^2\left({1\over e^{-\delta q_i}}+{1\over 1-e^{-\delta q_i}}\right)\right) -1
&& \tag{Since $F_{\delta}(P^{yes}_i)\sim \bern(1-e^{-\delta p_i})$ and $F_{\delta}(Q_i)\sim \bern(1-e^{-\delta q_i})$}\\
&= \exp(\sum_i (e^{-\delta p_i}-e^{-\delta q_i})^2/e^{-\delta q_i}(1-e^{-\delta q_i})) -1\\
&= \exp(\sum_i (e^{\delta (q_i-p_i)}-1)^2/(e^{\delta q_i}-1)) -1\\
&= \exp(\sum_i (e^{\delta (q_i-p_i)}-1)^2/(e^{\delta q_i}-1)) -1\\
&\le \exp(\sum_i (2\delta (p_i-q_i))^2/\delta q_i)) -1
&&\tag{Since $(e^x-1)\ge x$ and $(|e^{x}-1|\le 2|x|$ for $0<|x|<1$}\\
&= \exp(4\delta \chi^2(P,Q))-1\\
\end{align*}
\end{proof}
}

We are set to present the main lower bound result of this section.

\ssLB*

\begin{proof}
We start with the hard distributions $P^{yes},P^{no}$ and $Q$ from \cref{thm-unstr-chisq2sample}.  Then $\dchisqr(P^{yes},Q)\le \epsilon^2/216$ and $\dtv(P^{no},Q)\ge \epsilon$ for some constant $0<\epsilon<1$.
We apply the reduction of \cref{def-fdeltap} with $\delta=1/3$ to these three distributions. Then from \cref{lem-dtv-pfdeltap} and \cref{lem-chisqFdeltaPQ} we get the following two inequalities:
\begin{itemize}
\item $\dchisqr(F_{\delta}(P^{yes}),F_{\delta}(Q))\le \exp(4\chi^2(P^{yes},Q)/3)-1 < \epsilon^2/160$.
\item $\dtv(F_{\delta}(P^{no}),F_{\delta}(Q))> (1/3e^{1/3})\epsilon$.
\end{itemize}
It follows if we can distinguish $\dchisqr(F_{\delta}(P^{yes}),F_{\delta}(Q))\le \epsilon^2/160$ versus $\dtv(F_{\delta}(P^{no}),F_{\delta}(Q))> (1/3e^{1/3})\epsilon$, then we are able to decide the hard instance of \cref{thm-unstr-chisq2sample}. Moreover, in order to simulate each sample from the distribution $F_{1/2}(P)$, we need $\poi(1/2)$ samples from $P$. So, if we need $m$ samples in total, from the additive property of the Poisson distribution, we need $\poi(m/2)=O(m)$ samples from $P$ in total, except for $\exp(-m)$ probability. It follows, if we can decide the problem given in the theorem statement in $o(n/\log n)$ samples, we can decide the hard problem of \cref{thm-unstr-chisq2sample} in $o(n/\log n)$ samples as well. This leads to a contradiction. Replacing the constant $\epsilon$ by $3e^{1/3}\epsilon_1$, we get \cref{2S-LB}. 
\end{proof}

\subsection{Hardness of $\dkl$-vs-$\dtv$ Tolerant Identity Testing}
In this section we show that for an unknown product distribution $P$ and a known product distribution $Q$ over $\{0,1\}^n$, distinguishing $\dkl(P,Q)\le \epsilon^2/9$ versus $\dtv(P,Q)> \epsilon$, for a constant $\epsilon$, cannot be decided in general with a truly sublinear sample complexity. We use a reduction to the following  hardness result, for identity testing of unstructured distributions over $n$ items under $\mathrm{KL}$-tolerance, given in~\citet*{Daskalakis:2018:DDS:3174304.3175479}. We restate the theorem with changes in the constants. For a probability distribution $P=\langle p_1,p_2,\dots, p_n\rangle$ over $n$ items, $||P||_2^2=\sum_i p_i^2$.

\begin{restatable}{theoremRe}{unstrKLonesample}\label{thm-unstr-KL1sample}
There exists a constant $0<\epsilon<1$ and three distributions $P^{yes}, P^{no}$ and $Q$, each over the sample space $[n]$ such that: (1) $\dkl(P^{yes},Q)\le \epsilon^2/216$, whereas $\dtv(P^{no},Q)\ge \epsilon$, (2)  $||P^{yes}||_2^2=O(\log^2 n/n)$, and (3) given only sample accesses to one of $P^{yes}$ or $P^{no}$, and complete knowledge of $Q$, distinguishing $P^{yes}$ versus $P^{no}$ with probability $>4/5$, requires $\Omega(n/\log n)$ samples.
\end{restatable}
\begin{proof}
The proof of this Theorem appears in~\citet*{Daskalakis:2018:DDS:3174304.3175479} (in Theorem 6.2 of this version), except the fact $||P^{yes}||_2^2=O(\log^2 n/n)$ is not explicitly claimed. We prove this claim in the following, by observing from the original construction given in the paper by~\citet*{ValiantV10}.

The hard distribution $P^{yes}$ is the distribution $p^{-}_{\log k,\phi}$
as defined in Definition 12 of~\citet*{ValiantV10}. We use the following facts about this distribution $p^{-}_{\log k,\phi}$, given in Fact 11, Definition 12 and (in the end of the second paragraph in the proof of) Lemma 13 in~\citet*{ValiantV10}:
\begin{itemize}
\item $\phi$ is a small enough constant
\item The support size $n$ and the parameter $k$ are related as $n=32k\log k/\phi$
\item The `un-normalized' mass at each point is $x/32k$, where $j=\log k$ and $x\le 4j$
\item The `normalizing constant' $c_2$ (which makes the probability values sum up to 1) is at most $\phi/j$ where $j=\log k$
\end{itemize}
From these facts we conclude each probability mass is $c_2\cdot x/32k\le \phi/8k$, where $n=32k\log k/\phi$ for some constant $\phi$. Hence, $||P^{yes}||_2^2 \le \phi^2/64k^2\cdot 32k\log k/\phi=\phi\log k/2k=O(\log^2 n/n)$.
\end{proof}

We use the reduction given in \cref{def-fdeltap}. We establish the following lemma, relating KL distances between the original and the reduced distributions.  

\begin{restatable}{lemmaRe}{KLFdelta}\label{lem-KLFdeltaPQ}
$\dkl(F_{\delta}(P),F_{\delta}(Q))\le \left(\delta+{\delta^2\over 2}\right) \dkl(P,Q)+ {3\delta^2\over 2} ||P||_2^2$, for any $0<\delta\le 1$.
\end{restatable}
\begin{proof}
We use the following fact about the $\mathrm{KL}$-distance between two product distributions.
\begin{factRe}\label{lem-kl-additive}
For two distributions $P=\prod_{i=1}^n P_i$ and $Q=\prod_{i=1}^nQ_i$
over the same discrete sample space, it holds that
$\dkl(P,Q)= \sum_{i=1}^n \dkl(P_i,Q_i)$.
\end{factRe}
From \cref{fact-fdeltap}, both $F_{\delta}(P)$ and $F_{\delta}(Q)$ are product distributions, the 
distribution of the $i$-th component being $F_{\delta}(P_i)$ and $F_{\delta}(Q_i)$ respectively. Let $P=\langle p_1,p_2,\dots,p_n\rangle$ and $Q=\langle q_1,q_2,\dots,q_n\rangle$. Then $F_{\delta}(P_i)\sim \bern(1-e^{-\delta p_i})$ and $F_{\delta}(Q_i)\sim \bern(1-e^{-\delta q_i})$. 
\begin{align*}
\dkl & (F_{\delta}(P),F_{\delta}(Q))\\&=\sum_i \dkl(F_{\delta}(P_i),F_{\delta}(Q_i))\\
&=\sum_i \left[(1-e^{-\delta p_i})\ln \left({1-e^{-\delta p_i}\over 1-e^{-\delta q_i}}\right)+e^{-\delta p_i}\ln {e^{-\delta p_i} \over e^{-\delta q_i}}\right]\\
&=\sum_i \ln \left({1-e^{-\delta p_i}\over 1-e^{-\delta q_i}}\right)+\sum_i e^{-\delta p_i} \ln {
e^{-\delta p_i} (1-e^{-\delta q_i}) \over
e^{-\delta q_i} (1-e^{-\delta p_i})
}\\
&=\sum_i \ln { e^{\delta q_i}(e^{\delta p_i}-1)\over e^{\delta p_i}(e^{\delta q_i}-1)}+
\sum_i e^{-\delta p_i} \ln \left({ e^{\delta q_i}-1\over e^{\delta p_i}-1}\right)\\
&=\sum_i \ln {e^{\delta q_i} \over e^{\delta p_i}} + \sum_i \ln \left({ e^{\delta p_i}-1\over e^{\delta q_i}-1}\right)+\sum_i e^{-\delta p_i} \ln \left({ e^{\delta q_i}-1\over e^{\delta p_i}-1}\right)\\
&=\sum_i (q_i-p_i) + \sum_i (1-e^{-\delta p_i})\ln \left({ e^{\delta p_i}-1\over e^{\delta q_i}-1}\right)\\
&=\sum_i (1-e^{-\delta p_i})\ln \left({ e^{\delta p_i}-1\over e^{\delta q_i}-1}\right) &&\tag{Since $\sum_i p_i=\sum_i q_i = 1$}\\
&= \sum_{p_i>q_i} (1-e^{-\delta p_i})\ln \left({ e^{\delta p_i}-1\over e^{\delta q_i}-1}\right) + \sum_{q_i>p_i} (1-e^{-\delta p_i})\ln \left({ e^{\delta p_i}-1\over e^{\delta q_i}-1}\right)\\
&\le \sum_{p_i>q_i} \delta p_i\ln \left({ e^{\delta p_i}-1\over e^{\delta q_i}-1}\right) + \sum_{q_i>p_i} \left(\delta p_i -{1\over 2}\delta^2p_i^2\right)\ln \left({ e^{\delta p_i}-1\over e^{\delta q_i}-1}\right)\\
&=\sum_i \delta p_i\ln \left({ e^{\delta p_i}-1\over e^{\delta q_i}-1}\right) + \sum_{q_i>p_i} {\delta^2 p_i^2\over 2} \ln \left({ e^{\delta q_i}-1\over e^{\delta p_i}-1}\right)\\ 
&\le \sum_i \delta p_i\ln \left({\delta p_i (1+\delta p_i )\over \delta q_i}\right) + \sum_{q_i>p_i} {\delta^2 p_i^2\over 2} \ln \left(\delta q_i (1+\delta q_i)\over \delta p_i\right)\\ 
&=\delta\left(\sum_i p_i \ln {p_i \over q_i} + \sum_i p_i \ln (1+\delta p_i)\right) + {\delta^2\over 2}\left( \sum_{q_i > p_i} p_i^2 \ln {q_i \over p_i} + \sum_{q_i > p_i} p_i^2 \ln (1+\delta q_i)\right)\\
&\le \left(\delta+{\delta^2\over 2}\right) \sum_i p_i \ln {q_i \over p_i} + {3\delta^2\over 2} \sum_i p_i^2\\
&=\left(\delta+{\delta^2\over 2}\right) \dkl(P,Q)+ {3\delta^2\over 2} ||P||_2^2.
\end{align*}
\end{proof}

Now we present the lower bound for closeness testing of product distributions. 

\sLB*

\begin{proof}
We start with the distributions $P^{yes},P^{no}$ and $Q$ from the hardness result \cref{thm-unstr-KL1sample}. Then $\dkl(P^{yes},Q)\le \epsilon^2/216$, $||P^{yes}||_2^2=O(\log^2 n/n)$ and $\dtv(P^{no},Q)\ge \epsilon$ for some constant $0<\epsilon<1$.
We apply the reduction of \cref{def-fdeltap}, with $\delta=1/3$ to these three distributions. Then from \cref{lem-dtv-pfdeltap} and \cref{lem-KLFdeltaPQ} we get the following two:
\begin{itemize}
\item $\dkl(F_{\delta}(P^{yes}),F_{\delta}(Q))\le \epsilon^2/160$, for any large enough $n$.
\item $\dtv(F_{\delta}(P^{no}),F_{\delta}(Q))> (1/3e^{1/3})\epsilon$.
\end{itemize}
It follows if we can distinguish $\dkl(F_{\delta}(P^{yes}),F_{\delta}(Q))\le \epsilon^2/160$ versus $\dtv(F_{\delta}(P^{no}),F_{\delta}(Q))> (1/3e^{1/3})\epsilon$, then we will be able to decide the hard instance of \cref{thm-unstr-chisq2sample}. Moreover, in order to simulate each sample from the distribution $F_{1/2}(P)$, we need $\poi(1/2)$ samples from $P$. So, if we need $m$ samples in total, from the additive property of the Poisson distribution, we need $\poi(m/2)=O(m)$ samples from $P$ in total, except for $\exp(-m)$ probability. It follows, if we can decide the problem given in the Theorem statement in $o(n/\log n)$ samples, we can decide the hard problem of \cref{thm-unstr-KL1sample} in $o(n/\log n)$ samples as well. This leads to a contradiction. Replacing the constant $\epsilon$ by $3e^{1/3}\epsilon_1$, we get \cref{1S-LB}.
\end{proof}

Before moving on to the next section, we note that recently the question of $\dtv$-versus-$\dtv$ tolerant testing problem for uniformity testing of distributions over $[n]$ was settled to be $\Theta({n\over \log n}{1\over \epsilon^2})$ by \citet*{DBLP:journals/tit/JiaoHW18}. In particular, this gives a stronger guarantee for \cref{thm-unstr-chisq2sample} and \cref{thm-unstr-KL1sample} when $\epsilon$ is not a constant. This directly strengthens our \cref{2S-LB}; also \cref{1S-LB} whenever $||P_{yes}||_2^2=O(\epsilon^2)$, giving us a $\Omega({n\over \log n}{1\over \epsilon^2})$ lower bound for any $\epsilon$.

\subsection{Hardness of non-tolerant $\dtv$ Identity Testing for General Alphabets} \citet*{DaskalakisDK18} and \citet*{DBLP:conf/colt/CanonneDKS17} have given optimal lower bounds for non-tolerant testing w.r.t. $\dtv$ distance when $\lvert\Sigma\rvert=2$. In this section, we generalize their result for $\Sigma>2$ case and get an optimal lower bound in regard to \cref{1S-UB}.
We show the following theorem, generalizing the proof of \citet*{DaskalakisDK18} specifically.
\thmMain*

Our hard distributions are as follows:

$P=$ the uniform distribution over $[\ell]^n$.

$Q=$ a radom distribution from the mixture $\left\{\left\{{1\over \ell}\left(1\pm{\epsilon\over \sqrt{n}}\right)\right\}^{{\ell\over 2}}\right\}^n$. Each distribution of the mixture is a product distribution, whose $i$-th component is a distribution over $[\ell]$, which randomly assigns probability values either ${1\over \ell}\left(1+{\epsilon\over \sqrt{n}}\right),{1\over \ell}\left(1-{\epsilon\over \sqrt{n}}\right)$ or ${1\over \ell}\left(1-{\epsilon\over \sqrt{n}}\right),{1\over \ell}\left(1+{\epsilon\over \sqrt{n}}\right)$, based on a random vector from $\{0,1\}^{{\ell\over 2}}$, for every consecutive sample space items from $[\ell]$.

First claim we show is that each member of the mixture is $\Theta(\epsilon)$ far from $P$ in $\dtv$ distance.

\begin{claimRe}\label{claim:far}
Let $Q^*$ be any member of $Q$. Then $\dtv(P,Q^*)\ge\Theta(\epsilon)$.
\end{claimRe}
\begin{proof}
Note that all members of the mixture $Q$ are permutations of each other. Since $P$ is fixed to the uniform distribution, all of them have the same $\dtv$ to $P$. We fix $Q^*$ to be the distribution from $Q$, corresponding to $\{0,1\}^{\ell/2}$ at every component.

It is a known fact that applying a common function to the sample space items can only reduce $\dtv$. We apply the function which is parity of $x\in [\ell]$ component wise. Resulting sample space becomes $\{0,1\}^n$, $Q^*$ becomes $\bern\left(1+{\epsilon\over \sqrt{n}}\right)^n$ and $P$ becomes $\bern\left({1\over 2}\right)^n$. It is a standard fact that the $\dtv$ of the later pair is at least $\Theta(\epsilon)$ (see eg. \cite*{DBLP:conf/colt/CanonneDKS17}).
\end{proof}

Next we show that distingishing $k$ samples from $P$ and $Q$ is hard. Let $P^{\otimes k},Q^{\otimes k}$ be their distributions. Noting that the components of both the mixture and the uniform distribution are independent and symmetric. We can upper bound them by $n$ copies of the first component's distribution, using Pinsker's inequality and linearity of KL. 

\begin{align}\label{eqn:pinsker}
\dtvsqr(P^{\otimes k},Q^{\otimes k}) \lesssim \dkl(Q^{\otimes k},P^{\otimes k}) \le n\cdot\dkl(Q_1^{\otimes k},P_1^{\otimes k})
\end{align}

Computing $\dhelsqr$ or $\dkl$ are hard for the multinomial unlike \citet*{DaskalakisDK18} (cf. Lemma 17). So, we use a reduction to simplify the calculations.

Recall $P_1^{\otimes k}$ is the distribution of the first $k$ samples when $P$ is the uniform distribution over $[\ell]$ and $Q_1^{\otimes k}$ is the same when we take a random distribution from the $2^{{\ell\over 2}}$ size mixture.

We reduce $P_1$ to the distribution $\cP_1=\ber({1\over \ell})^l$ and any distribution from the mixture $Q^*_1=\langle q_1, \dots,q_{\ell}\rangle$ to $\cQ^*_1=\ber(q_1)\times\dots\times \ber(q_{\ell})$. We claim that this reduction changes KL by a constant factor, for any particular (randomly) chosen pair $P_1,Q^*_1$ to start with.

\begin{claimRe}\label{claim:near}
$\dhel(Q^*_1,P_1)\le \dkl(\cQ^*_1,\cP_1)$ and $\dkl({Q_1}^{\otimes k},P_1^{\otimes k})\le 6\dkl({\cQ_1}^{\otimes k},{\cP_1}^{\otimes k})$. 
\end{claimRe}
\begin{proof}
\begin{align*}
\dkl(Q_1,P_1)&=\sum_j Q_{ij}\log {Q_{1j}\over P_{1j}}\\
&=\sum_{j:Q_{1j}>{1\over \ell}} {1\over \ell}\left(1+{\epsilon\over \sqrt{n}}\right) \log \left(1+{\epsilon\over \sqrt{n}}\right)+\sum_{j:Q_{1j}<{1\over \ell}} {1\over \ell}\left(1-{\epsilon\over \sqrt{n}}\right) \log \left(1-{\epsilon\over \sqrt{n}}\right)\\
&\le {1\over 2}\left(1+{\epsilon\over \sqrt{n}}\right) {\epsilon\over \sqrt{n}}+ {1\over 2}\left(1-{\epsilon\over \sqrt{n}}\right) \left(-{\epsilon\over \sqrt{n}}\right)\\
&={\epsilon^2\over n}
\end{align*}

\begin{align*}
\dkl(\cQ_1,\cP_1)&={\ell\over 2}\dkl\left(\ber\left({1\over \ell}\left(1+{\epsilon\over \sqrt{n}}\right),\ber\left({1\over\ell}\right)\right)\right)+\\
&{\ell\over 2}\dkl\left(\ber\left({1\over \ell}\left(1-{\epsilon\over \sqrt{n}}\right)\right),\ber\left({1\over\ell}\right)\right)\\
\end{align*}
\begin{align*}
&\dkl\left(\ber\left({1\over \ell}\left(1+{\epsilon\over \sqrt{n}}\right),\ber\left({1\over\ell}\right)\right)\right)\\
&={1\over \ell}\left(1+{\epsilon\over \sqrt{n}}\right)\log (1+{\epsilon\over \sqrt{n}})+\left(1-{1\over \ell}-{\epsilon\over \ell\sqrt{n}}\right)\log \left(1-{\epsilon\over (\ell-1)\sqrt{n}}\right)\\
&\ge {1\over \ell}\left(1+{\epsilon\over \sqrt{n}}\right)\left({\epsilon\over \sqrt{n}}-{\epsilon^2\over 2n}\right)+\left(
1-{1\over\ell}-{\epsilon\over \ell\sqrt{n}}
\right)\left(-{\epsilon\over \sqrt{n}(\ell-1)}-{\epsilon^2\over n(\ell-1)^2}\right)\\
&\ge{\epsilon^2\over 3n\ell}
\end{align*}
Therefore, $\dkl(\cQ_1,\cP_1) \ge \dkl(Q_1,P_1)/6$. Using linearity of KL, we get $\dhel({Q_1^*}^{\otimes k},P_1^{\otimes k})\le 6\dkl({\cQ_1^*}^{\otimes k},\cP_1^{\otimes k})$. Since this holds for the reduction on any chosen starting pair, we get that in general for the mixture, 
\[\dkl({Q_1}^{\otimes k},P_1^{\otimes k})\le 6\dkl({\cQ_1}^{\otimes k},{\cP_1}^{\otimes k})\]
\end{proof}

Henceforth we focus on upper bounding $\dkl(\cQ_1^{\otimes k},\cP_1^{\otimes k})$. The reduction makes it a Boolean product distribution. \citet*{DaskalakisDK18} gave such upper bounds when every component is randomly mixed and when the probabilities are close to $1/2$. Instead we need to mix every pairs of components and need to make the probabilities close to $1/\ell$.

Let $p_+={1\over \ell}\left(1+{\epsilon\over\sqrt{n}}\right)$ and $p_-={1\over \ell}\left(1-{\epsilon\over\sqrt{n}}\right)$.

\begin{proof} (of \cref{thm:main})
We firstly note that it suffices to upper bound the joint distribution of the count of 1's in the samples $\cQ_1^{\otimes k}$ and $\cP_1^{\otimes k}$ (\citet*{DaskalakisDK18}, Lemma 17).
By symmetry, we can focus on the mixture on the first two components $1,2$. Recall these are actually the first two of the $\ell$ components of the reduced distribution (where the reduction was performed on the first component of the original distribution). Note that \citet*{DaskalakisDK18} could instead focus on a single component.

\begin{align}\label{eqn:12}
\dkl(Q_1^{\otimes k},P_1^{\otimes k})\lesssim \dkl(\cQ_1^{\otimes k},\cP_1^{\otimes k})={\ell\over 2}\dkl({\cQ_1^{\otimes k}}_{12},{\cP_1^{\otimes k}}_{12})\le {\ell\over 2}\dchisqr(R^{\otimes k}_{12},S^{\otimes k}_{12})
\end{align}
 
We use $R=\cQ_1$ and $S=\cP_1$ for simplicity.
Then, $R^{\otimes k}_{12}$ and $S^{\otimes k}_{12}$ denote the joint distribution of the count of 1s in $k$ samples at the first 2 components w.r.t the distributions $\ber\left({1\over \ell}\right)$ and the $2^{\ell/2}$-sized mixture $\ber\left({1\over \ell}\left(1\pm{\epsilon\over \sqrt{n}}\right)\right)$ (the former and the later are due to our reduction).

\begin{align*}
&1+\dchisqr(R^{\otimes k}_{12},S^{\otimes k}_{12})\\
&=\sum_{i=0}^k\sum_{j=0}^k {\left[
{1\over 2}{k\choose i}(p_+)^i(1-p_+)^{k-i} + {1\over 2}{k\choose j}(p_-)^j(1-p_-)^{k-j}
\right]^2
\over 
{k\choose i} \left({1\over \ell}\right)^i \left(1-{1\over \ell}\right)^{k-i}
{k\choose j} \left({1\over \ell}\right)^j \left(1-{1\over \ell}\right)^{k-j}
}\\
&={\ell^{2k}\over 4}\sum_{i=0}^k\sum_{j=0}^k {k\choose i}{k\choose j} [
\left(p_+^2\right)^i \left({(1-p_+)^2\over(\ell-1)}\right)^{k-i}
\left(p_-^2\right)^j \left({(1-p_-)^2\over(\ell-1)}\right)^{k-j} +\\
&\qquad\qquad\qquad\qquad\qquad\left(p_-^2\right)^i \left({(1-p_-)^2\over(\ell-1)}\right)^{k-i}
\left(p_+^2\right)^j \left({(1-p_+)^2\over(\ell-1)}\right)^{k-j} +\\
&\qquad\qquad\qquad\qquad\qquad2\left(p_+p_-\right)^i \left({(1-p_+)(1-p_)\over (\ell-1)}\right)^{k-1}\left(p_-p_+\right)^j \left({(1-p_)(1-p_+)\over (\ell-1)}\right)^{k-j}
]\\
&={\ell^{2k}\over 4}\left[
2\left(p_+^2+{(1-p_+)^2\over \ell-1}\right)^k\left(p_-^2+{(1-p_-)^2\over \ell-1}\right)^k+2\left(p_+p_-+{(1-p_+)(1-p_-)\over (\ell-1)}\right)^{2k}
\right]
\end{align*}

\begin{align*}
p_+^2+{(1-p_+)^2\over (\ell-1)}&={1\over \ell^2}\left(1+{\epsilon\over \sqrt{n}}\right)^2+{1\over (\ell-1)}\left(1-{1\over \ell-{\epsilon\over l\sqrt{n}}}\right)^2
&={1\over \ell}+{\epsilon^2\over n\ell(\ell-1)} &&\tag{upon simplification}
\end{align*}

\begin{align*}
p_+^2+{(1-p_+)^2\over (\ell-1)}&= {1\over \ell}+{\epsilon^2\over n\ell(\ell-1)}&&\tag{upon simplification}
\end{align*}

\begin{align*}
\left(p_+p_-+{(1-p_+)(1-p_)\over (\ell-1)}\right)&={1\over \ell}-{\epsilon^2\over n\ell(\ell-1)}&&\tag{upon simplification}
\end{align*}

Therefore, 
\begin{align*}
1+\dchisqr(R^{\otimes k}_{12},S^{\otimes k}_{12})
&\le {\ell^{2k}\over 2}\left[\left({1\over \ell}+{\epsilon^2\over n\ell(\ell-1)}\right)^{2k}+\left({1\over \ell}-{\epsilon^2\over n\ell(\ell-1)}\right)^{2k}\right]\\
&\approx 1+{2k\choose 2}\left(\epsilon^2 \over n(\ell-1)\right)^2+\dots
\end{align*}

We get  that if $k=o(\sqrt{nl}\epsilon^{-2})$, then $\dchisqr(R^{\otimes k}_{12},S^{\otimes k}_{12})=o({1\over nl})$. Then we get from \cref{eqn:pinsker} and \cref{eqn:12}, $\dtv(P^{\otimes k},Q^{\otimes k})=o(1)$, establishing \cref{thm:main}. 
\end{proof}


\section{Tolerant Testing in $\dhel$ for High-Dimensional Distributions}
In this section, we give an algorithm for tolerant testing of two high-dimensional distributions w.r.t. the Hellinger distance. More specifically, given samples from such a pair of unknown distributions $P$ and $Q$ our goal would be to distinguish between the cases: $\dhel(P,Q) \le \epsilon/2$ versus $\dhel(P,Q) > \epsilon$ with 2/3 probability, which can be amplified to $1-\delta$  using the majority of $O(\log {1\over \delta})$ repetitions, for any $0<\delta,\epsilon<1$. This generalizes the work of \citet*{bhattacharyya2020efficient}, who gave such tolerant testers w.r.t. $\dtv$ using distance approximation. 

We start with a result that additively estimates $\dhelsqr(P,Q)$, when we have access to both the p.m.f.s and also to independent samples from $P$.

\begin{theoremRe}\label{thm:distEst}
Consider $P$ and $Q$ be two unknown distributions over $\Sigma^n$. Suppose we have access to two circuits $\xi_P(x)$ and $\xi_Q(x)$ which on input $x$, outputs $P(x)$ and $Q(x)$ respectively. Then we can output a number $e$ such that $|e-\dhelsqr(P,Q)| \le \epsilon$ with 2/3 probability for any $0<\epsilon<1$, using $ 3\epsilon^{-2}$ independent samples from $P$ and $3\epsilon^{-2} $ calls to each of $\xi_P(x)$ and $\xi_Q(x)$.
\end{theoremRe}

\begin{proof}
\begin{align*}
1-\dhelsqr(P,Q)&=\sum_{x\in \Sigma^n} \sqrt{P(x)Q(x)}\\
&=\sum_{x\in \Sigma^n} \sqrt{Q(x)\over P(x)}P(x)\\
&=\e_{x\sim P} \left[\sqrt{Q(x)\over P(x)}\right]&&\tag{since $P(x)\neq 0$}
\end{align*}

Let $f(x) = \sqrt{Q(x)\over P(x)}$. 
Therefore, it suffices to estimate $\e_{x\sim P} \left[f(x)\right]$ additively.
Note that, $\var_{x\sim P}[f(x)] \le \e_{x\sim P} \left[f^2(x)\right]\ab=\sum_x Q(x)=1$. We define our estimator to be $e$, the average of $(1-f(x))$ over $R$ samples from $P$. Then $e$ satisfies $\e[e]=\dhelsqr(P,Q)$ and $\var[e]\le 1/R$. Chebyshev's inequality gives us that for $R\ge 3\epsilon^{-2}$, $|e-\dhelsqr(P,Q))|\le \epsilon$ with at least 2/3 probability. 
\end{proof}
\subsection{Application: Bayesian Networks}
\citet*{bhattacharyya2020efficient} have given the following Algorithm for learning an unknown Bayesian network on a known graph of indegree at most $d$.

\begin{restatable}{theoremRe}{bnlearn}\label{thm:bnlearning}
There is an algorithm that given a parameter $\epsilon>0$ and sample access to an unknown Bayesian network distribution $P$ on a known directed acyclic graph $G$ of in-degree at most $d$, returns a Bayesian network $\hat{P}$ on $G$ such that $\dhel(P,\hat{P}) \leq \epsilon$ with probability $\geq 9/10$. Letting $\Sigma$ denote the range of each variable $X_i$, the algorithm takes $m=O(|\Sigma|^{d+1}n\log(|\Sigma|^{d+1} n)\epsilon^{-2})$ samples and runs in $O(mn)$ time.  
\end{restatable}

We get the following result for tolerant resting of Bayesian networks in Hellinger distance.
\bnTolH*
\begin{proof}
First we learn $P$ and $Q$ using \cref{thm:bnlearning} such that $\dhel(P,\hat{P})\le \epsilon/12$ and $\dhel(Q,\hat{Q})\le \epsilon/12$. This step costs $m=O(|\Sigma|^{d+1}n\log(|\Sigma|^{d+1} n)\epsilon^{-2})$ samples, runs in $O(mn)$ time, and succeeds with 4/5 probability. Note that $\hat{P}$ and $\hat{Q}$, once learnt, can be sampled and evaluated correctly in $O(n)$ time.

Next we estimate $\dhelsqr(\hat{P},\hat{Q})$ up to an additive $\epsilon^2/9$ error using \cref{thm:distEst}. This step costs $O(n\epsilon^{-4})$ time and no further samples and succeeds with 4/5 probability. Due to the triangle inequality of $\dhel$, in the first case, $\dhelsqr(\hat{P},\hat{Q})\le 20\epsilon^2/36$ and in the second case $\dhelsqr(\hat{P},\hat{Q})>21\epsilon^2/36$, thus separating the two cases.
\end{proof}

\section{Non-tolerant Closeness Testers}\label{sec:nont-2sample}
\subsection{$\dhel$-tester}\label{sec:inHelDistance}
A non-tolerant tester for 2-sample testing of product distributions was given in~\citet*{DBLP:conf/colt/CanonneDKS17}. Their tester distinguishes $P=Q$ from
$\dtv(P,Q)\ge \epsilon$ with sample complexity $O(\max\{n^{3/4}/\epsilon,\ab\sqrt{n}/\epsilon^2\})$. Using the relation, $\dtv \ge \dhelsqr$ from \cref{lem-distcompare}, we immediately get a tester for distinguishing $P=Q$ from $\dhel(P,Q)\ge \epsilon$, with sample complexity $O(\max\{n^{3/4}/\epsilon^2,\sqrt{n}/\epsilon^4\})$. Here, we show an improved tester with $O(n^{3/4}/\epsilon^2)$ sample complexity in \cref{algo-nontolerant-hellinger-2sample}. We analyze its correctness and complexity below.

\begin{algorithm2e}
\caption{Given samples from two unknown distributions $P=P_1\times \dots \times P_n$ and $Q=Q_1 \times \dots \times Q_n$ over $\Sigma^n$, decides $P=Q$ (`yes') versus $\dhel(P,Q)\ge \epsilon$ (`no'). Let $\ell=|\Sigma|$.}
\label{algo-nontolerant-hellinger-2sample}
\tcc{Approximately identify heavy and light partitions}
Take $m$ samples from $P$ and $Q$. Let $U'\subseteq [n]\times [\ell]$ be the set of indices $(i,j)$, such that at least one sample from either
$P$ or $Q$ has hit symbol $j\in \Sigma$ in the coordinate $i$\;
Let $V'=[n]\times [\ell] \setminus U'$\;
\DontPrintSemicolon
\;
\tcc{Poisson sampling}
For each $i\in [n]$, sample $M_i \sim \poi(m)$ independently \;
For each $i\in [n]$, sample $M'_i \sim \poi(m)$ independently \;
Let $M=\max_i\{M_i\}$ and $M'=\max_i\{M'_i\}$\;
If $\max\{M,M'\}\ge 2m$ output `no'\;
Take $M$ samples $X^{1},\dots,X^{M}$ from $P$\;
Take $M'$ samples $Y^{1},\dots,Y^{M'}$ from $Q$\;
For every $(i,j)$, let $W_{ij}$ be the number of occurrences of symbol $j\in \Sigma$ in the $i$-th coordinate of the sample subset $X^{1},\dots,X^{M_i}$\;
For every $(i,j)$, let $V_{ij}$ be the number of occurrences of symbol $j\in \Sigma$ in the $i$-th coordinate of the sample subset$Y^{1},\dots,Y^{M'_i}$\;
\;

\tcc{Test the heavy partition}
$W_{heavy}=\sum_{(i,j) \in U'}{(W_{ij}-V_{ij})^2-(W_{ij}+V_{ij}) \over (W_{ij}+V_{ij})}$\;
If $W_{heavy}>m\epsilon^2/120$ output `no'\;
\;

\tcc{Test the light partition}
$W_{light}=\sum_{(i,j) \in V'} (W_{ij}-V_{ij})^2-(W_{ij}+V_{ij})$\;
If $W_{light}>m^2\epsilon^4/1000n\ell$ output `no'\;
Output `yes'\;
\end{algorithm2e}

Let $P=\prodp$ and $Q=\prodq$, with $P_i=\langle p_{i1},\dots,p_{i\ell}\rangle$ and $Q_i=\langle q_{i1},\dots,q_{i\ell}\rangle$ as probability vectors. We assume $\min_{i,j} p_{ij}\ge \epsilon^2/50n\ell$ and  $\min_{i,j} q_{ij}\ge \epsilon^2/50n\ell$, without loss of generality using the reduction of \cref{lem-modification}. 

Analysis of \cref{algo-nontolerant-hellinger-2sample} can be divided into two cases: `heavy' and `light'. Let $V\subseteq [n]\times [\ell]$ be the `light' set of indices $(i,j)$, such that $\max\{p_{ij},q_{ij}\}<1/m$. The remaining indices in $U=[n]\times [\ell] \setminus V$ are `heavy'. The following important lemma shows that for each case, a certain sum must deviate from zero substantially, for the `no' class.
\begin{lemmaRe}\label{lem-dhelcases}
Suppose $\dhel(P,Q)\ge \epsilon$. Suppose $\min_{i,j} p_{ij}\ge \epsilon^2/50n\ell$ and  $\min_{i,j} q_{ij}\ge \epsilon^2/50n\ell$. Then at least one of the following two must hold:
\begin{enumerate}
\item $\sum_{(i,j)\in V} (p_{ij}-q_{ij})^2\ge \epsilon^4/25n\ell$
\item $\sum_{(i,j)\in U} {(p_{ij}-q_{ij})^2\over p_{ij}+q_{ij}}\ge \epsilon^2$.
\end{enumerate}
\end{lemmaRe}
\begin{proof}
If $\dhel(P,Q)\ge \epsilon$, \cref{lem-helloc-equality} gives us $\sum_i \dhelsqr(P_i,Q_i) \ge \epsilon^2$.  We use the following standard Fact to get $\sum_{i=1}^n \sum_{j=1}^\ell {(p_{ij}-q_{ij})^2\over p_{ij}+q_{ij}} \ge 2 \sum_i \dhelsqr(P_i,Q_i) \ge 2\epsilon^2$.
\begin{factRe}\label{fact-triangledist}(see eg. \citet*{Daskalakis:2018:DDS:3174304.3175479})
For two distributions $P=\{p_1,\dots,p_\ell\}$ and $Q=\{q_1,\dots,q_\ell\}$, $\sum_j {(p_j-q_j)^2\over p_j+q_j}\ge 2 \dhelsqr(P,Q)$.
\end{factRe}
It follows that at least one of 1) $\sum_{(i,j)\in V} {(p_{ij}-q_{ij})^2\over p_{ij}+q_{ij}}$ or
2) $\sum_{(i,j)\in U} {(p_{ij}-q_{ij})^2\over p_{ij}+q_{ij}}$  
is at least $\epsilon^2$.

In the first case, 
\begin{align*}
\sum_{(i,j)\in V} (p_{ij}-q_{ij})^2&=\sum_{(i,j)\in V}  (\sqrt{p_{ij}}-\sqrt{q_{ij}})^2(\sqrt{p_{ij}}+\sqrt{q_{ij}})^2\\
&\ge (2\epsilon^2/25n\ell) \cdot \sum_{(i,j)\in V}  (\sqrt{p_{ij}}-\sqrt{q_{ij}})^2 &&\tag{as $\min_{i,j} \min\{p_{ij}, q_{ij}\}\ge \epsilon^2/50n\ell$}\\
&\ge (2\epsilon^2/25n\ell) \cdot \sum_{(i,j)\in V} {(p_{ij}-q_{ij})^2 \over (\sqrt{p_{ij}}+\sqrt{q_{ij}})^2}\\
&\ge (\epsilon^2/25n\ell)\cdot \sum_{(i,j)\in V} {(p_{ij}-q_{ij})^2 \over p_{ij}+q_{ij}}\\
&\ge \epsilon^4/25n\ell.
\end{align*}
\end{proof}
The following lemma shows $U'$ ($V'$), as obtained in Lines 1-2 of \cref{algo-nontolerant-hellinger-2sample}, could be an acceptable proxy for $U$ ($V$).
\begin{lemmaRe}\label{lem:checkProxy}
Let $U',V'$ be as in \cref{algo-nontolerant-hellinger-2sample}. Let $m=\Omega(\sqrt{n\ell}/\epsilon^2)$ for some sufficiently large constant. Then with probability at least $0.63$ in each case, the following holds:
\begin{enumerate}
\item $\sum_{(i,j)\in U} {(p_{ij}-q_{ij})^2\over p_{ij}+q_{ij}}\ge \epsilon^2$ implies $\sum_{(i,j)\in U \cap U'} {(p_{ij}-q_{ij})^2\over p_{ij}+q_{ij}}\ge \epsilon^2/20$
\item $\sum_{(i,j)\in V} (p_{ij}-q_{ij})^2\ge \epsilon^4/25n\ell$ implies $\sum_{(i,j)\in V'} (p_{ij}-q_{ij})^2\ge \epsilon^4/500n\ell$
\end{enumerate}
\end{lemmaRe}
\begin{proof}
Note that $(i,j) \in V'$ with probability = $(1-p_{ij})^m(1-q_{ij})^m$, and $(i,j) \in U'$ with the remaining probability.

({\it Proof of 1:}) Let $U''=U \cap U'$. Suppose there exists $(i,j) \in U$ such that ${(p_{ij}-q_{ij})^2\over p_{ij}+q_{ij}}\ge \epsilon^2/20$. Then this $(i,j) \in U''$ with probability $1-(1-p_{ij})^m(1-q_{ij})^m \ge 1-(1-1/m)^m \ge 0.63$, in which case the result follows. Otherwise, we consider sum of the independent random variables, $S=\sum_{(i,j)} 1_{(i,j)\in U''} {20(p_{ij}-q_{ij})^2\over \epsilon^2(p_{ij}+q_{ij})}$, each of which is in $[0,1]$. $\mathrm{E}[S]=\sum_{(i,j)\in U} (1-(1-p_{ij})^m(1-q_{ij})^m) {20(p_{ij}-q_{ij})^2\over \epsilon^2(p_{ij}+q_{ij})} \ge 12.6$. We apply Chernoff's bound to get $S\ge 6.3$ with probability 0.63.

({\it Proof of 2:}) Let $V''=V \cap V'$. We consider sum of the independent random variables, $S=\sum_{(i,j)} 1_{(i,j)\in V''} m^2 (p_{ij}-q_{ij})^2$, each of which is in $[0,1]$. $\mathrm{E}[S]=\sum_{(i,j)\in V} (1-p_{ij})^m(1-q_{ij})^m m^2 (p_{ij}-q_{ij})^2 > (1-1/m)^{2m} \sum_{(i,j)\in V} m^2 (p_{ij}-q_{ij})^2 \ge m^2\epsilon^4/250n\ell$, for $m\ge 4$. We apply Chernoff's bound to get $S\ge m^2\epsilon^4/500n\ell$ except for probability at most $\exp(-m^2\epsilon^4/3000n\ell)$.
\end{proof}
Combining \cref{lem-dhelcases} and \cref{lem:checkProxy}, we get for the `no' case, one of the two conditions of \cref{lem:checkProxy} must hold. \cref{algo-nontolerant-hellinger-2sample} uses the two tests $W_{heavy}$ and $W_{light}$ to check these two conditions separately. To analyze them, we use certain important results from~\citet*{DBLP:conf/colt/CanonneDKS17}, assuming $W_{ij}\sim \poi(p_{ij})$ and $V_{ij}\sim \poi(q_{ij})$ for every $i,j$, which holds due to Poisson sampling. We assume the check of Line 7 goes through except 1/50 probability, using the concentration of Poisson distribution.
\paragraph{Analysis of $W_{heavy}$} 
\begin{lemmaRe}[Obtained from Claims 37 and 38 of~\citet*{DBLP:conf/colt/CanonneDKS17}]\label{lem:heavy1}
If $P=Q$ then $\e[W_{heavy}]=0$. If  $\sum_{(i,j)\in U\cap U'} {(p_{ij}-q_{ij})^2\over p_{ij}+q_{ij}}\ge \epsilon^2/20$ then $\e[W_{heavy}]\ge m\epsilon^2/60$.
In both cases $\var[W_{heavy}] \le 7n\ell+15\e[W_{heavy}]$
\end{lemmaRe}
By the application of Chebyshev's inequality we get the following.
\begin{lemmaRe}\label{lem:heavy2}
Let $m=\Omega(\sqrt{n\ell}/\epsilon^2)$ for a sufficiently large constant. Then the following holds
except for probability $\le 1/25$ in each case,
\begin{enumerate}
\item $P=Q$ implies $W_{heavy}\le m\epsilon^2/120$
\item $\sum_{(i,j)\in U\cap U'} {(p_{ij}-q_{ij})^2\over p_{ij}+q_{ij}}\ge \epsilon^2/20$ implies $W_{heavy}> m\epsilon^2/120$.
\end{enumerate}
\end{lemmaRe}

\paragraph{Analysis of $W_{light}$}
\begin{lemmaRe}[Obtained from Proposition 6 of~\citet*{ChanDVV14} and Claim 35 of~\citet*{DBLP:conf/colt/CanonneDKS17}]\label{lem:light1}
$\e[W_{light}]=m^2\sum_{(i,j)\in V'} {(p_{ij}-q_{ij})^2}$ and $\var[W_{light}] \le 80m^3\sqrt{b}\sum_{(i,j)\in V'} {(p_{ij}-q_{ij})^2}+8m^2b$, where $b=\max\{\sum_{(i,j) \in V'}p_{ij}^2,\sum_{(i,j) \in V'}q_{ij}^2\}$. Furthermore, $b\le 50n\ell/m^2$ for a sufficiently large $m$, except for probability at most $1/50$.
\end{lemmaRe}
By the application of Chebyshev's inequality we get the following.
\begin{lemmaRe}\label{lem:light2}
Let $m=\Omega((n\ell)^{3/4}/\epsilon^2)$ for a sufficiently large constant. Then the following holds
except for probability $\le 1/50$ in each case,
\begin{enumerate}
\item $P=Q$ implies $W_{light}\le m^2\epsilon^4/1000n\ell$
\item $\sum_{(i,j)\in V'} (p_{ij}-q_{ij})^2\ge \epsilon^4/500n\ell$ implies $W_{light}> m^2\epsilon^4/1000n\ell$.
\end{enumerate}
\end{lemmaRe}
Together we get $O((n\ell)^{3/4}/\epsilon^2)$ samples are enough to distinguish $P=Q$ versus $\dhel(P,Q)\ge \epsilon$. 

\clexhel*

\subsection{$\dtv$-tester}
A 2-sample tester for distinguishing $P=Q$ from
$\dtv(P,Q)\ge \epsilon$, for product distributions over $\{0,1\}^n$, was given in~\citet*{DBLP:conf/colt/CanonneDKS17} with sample complexity $O(\sqrt{n}/\epsilon^2,\max\{n^{3/4}/\epsilon\})$. In the following, we generalize this result for product distributions over $\Sigma^n$ with sample complexity $m=O(\max\{\sqrt{n|\Sigma|}/\epsilon^2,(n|\Sigma|)^{3/4}/\epsilon\})$ which is optimal for $|\Sigma|=2$ \citep*{DBLP:conf/colt/CanonneDKS17}. Let $\ell=|\Sigma|$. We assume $\min_{i,j} p_{ij}\ge \epsilon/50n\ell$ and  $\min_{i,j} q_{ij}\ge \epsilon/50n\ell$ without loss of generality, using a reduction similar to \cref{lem-modification} \citep*{DBLP:conf/colt/CanonneDKS17,DBLP:conf/colt/DaskalakisP17}. Let $V\subseteq [n]\times [\ell]$ be the set of indices $(i,j)$, such that $\max\{p_{ij},q_{ij}\}<1/m$. Let $U=[n]\times [\ell] \setminus V$.
\begin{lemmaRe}\label{lem-dtvcases}
Suppose $\dtv(P,Q)\ge \epsilon$. Suppose $\min_{i,j} p_{ij}\ge \epsilon/50n\ell$ and  $\min_{i,j} q_{ij}\ge \epsilon/50n\ell$. Then at least one of the following two must hold: 
\begin{enumerate}
\item $\sum_{(i,j)\in V} (p_{ij}-q_{ij})^2\ge \epsilon^2/n\ell$ 
\item $\ab\sum_{(i,j)\in U} {(p_{ij}-q_{ij})^2\over p_{ij}+q_{ij}}\ge \epsilon^2/4$.
\end{enumerate}
\end{lemmaRe}
\begin{proof}
We get $\sum_{(i,j)\in V} |p_{ij}-q_{ij}|+\sum_{(i,j)\in U} |p_{ij}-q_{ij}| = 2\sum_i \dtv(P_i,Q_i) \ge 2\dtv(P,Q) \ge 2\epsilon$, the second last inequality from the following Fact.
\begin{factRe}
For two product distributions $P=\Pi_{i=1}^n P_i$ and $Q=\Pi_{i=1}^n Q_i$, $\dtv(P,Q) \le \sum_i \dtv(P_i,Q_i)$. 
\end{factRe}
 Hence at least one of $\sum_{(i,j)\in V} |p_{ij}-q_{ij}|$ or $\sum_{(i,j)\in U} |p_{ij}-q_{ij}|$ is at least $\epsilon$.

In the first case, we get
\begin{align*}
\qquad\qquad\quad\sum_{(i,j)\in V} (p_{ij}-q_{ij})^2 \sum_{(i,j) \in V} 1 &\ge (\sum_{(i,j) \in V} |p_{ij}-q_{ij}|)^2 &&\tag{Cauchy-Schwarz inequality}\\
\sum_{(i,j)\in V} (p_{ij}-q_{ij})^2& \ge  \epsilon^2/n\ell.
\end{align*}

In the second case, the proof is similar to that of the standard Facts $\dtv\le \dhel$ and \cref{fact-triangledist} (see eg. \citet*{Daskalakis:2018:DDS:3174304.3175479} for both).
\begingroup
\allowdisplaybreaks
\begin{align*}
\qquad\qquad\epsilon^2&\le (\sum_{(i,j)\in U} |p_{ij}-q_{ij}|)^2\\
&= (\sum_{(i,j)\in U} |\sqrt{p_{ij}}-\sqrt{q_{ij}}||\sqrt{p_{ij}}+\sqrt{q_{ij}}|)^2\\
&\le (\sum_{(i,j)\in U} (\sqrt{p_{ij}}-\sqrt{q_{ij}})^2 )(\sum_{(i,j)\in U} (\sqrt{p_{ij}}+\sqrt{q_{ij}})^2 )&&\tag{Cauchy-Schwarz inequality}\\
&\le (\sum_{(i,j)\in U} (\sqrt{p_{ij}}-\sqrt{q_{ij}})^2 )(\sum_{(i,j)\in U}2(p_{ij}+q_{ij}))\\
&\le 4\sum_{(i,j)\in U} {(p_{ij}-q_{ij})^2 \over (\sqrt{p_{ij}}+\sqrt{q_{ij}})^2}\\
&\le 4\sum_{(i,j)\in U} {(p_{ij}-q_{ij})^2 \over p_{ij}+q_{ij}}.
\end{align*}
\endgroup

\end{proof}
We skip the rest of the details of the algorithm and its analysis since it closely follows that of \cref{sec:inHelDistance}. We identify the partitions $U$ and $V$ approximately by checking which indices are hit in $m$ samples. For $m=\Omega(\sqrt{n\ell}/\epsilon)$, this approximation is acceptable using a result similar to \cref{lem:checkProxy}. \cref{lem:heavy1}, \cref{lem:heavy2} and \cref{lem:light1} are as before up to the constants. Only in \cref{lem:light2}, the threshold for the light part changes to $m^2\epsilon^2/40n\ell$, and the sample complexity for the light part changes to $m=\Theta((n\ell)^{3/4}/\epsilon)$.
\nontolDTV*

\paragraph{Acknowledgement} We thank Clément Canonne for his comments on an earlier version of this paper. We also thank the anonymous reviewers of ALT 21 for improving the paper.

\bibliography{reflist}
\newpage
\appendix
\section{Proof of \cref{thm-adk}}
\adk*
\begin{proof}
The test $T$ of~\citet*{AcharyaDK15} is given by $T=\Sigma_{i=1}^K {(N_i-ms_i)^2 - N_i\over ms_i}$.
\begingroup
\allowdisplaybreaks
\begin{align*}
\e[T]&=\e\left[\sum_i {(N_i-ms_i)^2-N_i\over ms_i}\right]\\
&=\sum_i {\e[(N_i-ms_i)^2-N_i]\over ms_i}\\
&=\sum_i {\e[N_i^2]+m^2s_i^2-2ms_i\e[N_i]-\e[N_i]\over ms_i}\\
&=\sum_i {mr_i(1+mr_i)+m^2s_i^2-2ms_i\cdot mr_i-mr_i\over ms_i} &&\tag{Since $N_i\sim \poi(mr_i)$}\\
&=m\sum_i {(r_i-s_i)^2\over s_i}.
\end{align*}
\endgroup

\begingroup
\allowdisplaybreaks
\begin{align*}
\var[T]&=\var\left[\sum_i {(N_i-ms_i)^2-N_i\over ms_i}\right]\\
&=\sum_i \var\left[{(N_i-ms_i)^2-N_i\over ms_i}\right] &&\tag{Since $N_i$s are independent for different $i$s}\\
&=\sum_i {1\over m^2s_i^2}\var[N_i^2-(2ms_i+1)N_i] \\
&=\sum_i {1\over m^2s_i^2} [\var[N_i^2]+(2ms_i+1)^2\var[N_i]-2(2ms_i+1)\mathrm{Cov}(N_i^2,N_i)]\\
&=\sum_i {1\over m^2s_i^2} [(\e[N_i^4]-\e^2[N_i^2])+(2ms_i+1)^2(\e[N_i^2]-\e^2[N_i])\\&\qquad\qquad\qquad\qquad\qquad\qquad\qquad\qquad\qquad\qquad-2(2ms_i+1)(\e[N_i^3]-\e[N_i^2]\e[N_i])]\\
&=\sum_i {1\over m^2s_i^2} [\lambda(1+5\lambda+4\lambda^2)+\lambda(2ms_i+1)^2-2(2ms_i+1)(\lambda(2\lambda+1))]&&\tag{Since $N_i\sim \poi(\lambda)$, where $\lambda=mr_i$}\\
&=\sum_i {1\over m^2s_i^2} \lambda[\lambda+(2ms_i-2\lambda)^2]\\
&=\sum_i {r_i^2\over s_i^2}+\sum_i 4mr_i{(r_i-s_i)^2\over s_i^2}.
\end{align*}
\endgroup
We bound the above two summations separately.
\begingroup
\allowdisplaybreaks
\begin{align*}
\sum_i {r_i^2\over s_i^2} &= \sum_i {(r_i-s_i)^2+2s_i(r_i-s_i)+s_i^2\over s_i^2}\\
&= \sum_i {(r_i-s_i)^2\over s_i^2} + 2\sum_i {r_i-s_i\over s_i} +\sum_i 1\\
&\le 2\left( \sum_i {(r_i-s_i)^2\over s_i^2}+\sum_i 1 \right) &&\tag{Using $a^2+1\ge 2a$}\\
&\le 2\left({50K\over \epsilon^2}\sum_i {(r_i-s_i)^2\over s_i}+K\right) &&\tag{Using $s_i \geq \epsilon^2/50K$}\\
&=2({50K\over \epsilon^2}{\e[T]\over m}+K)\\
&\le \sqrt{K}\e[T]+2K. &&\tag{Using $m\ge c\sqrt{K}/\epsilon^2$ for c sufficiently large}
\end{align*}
\endgroup
\begingroup
\allowdisplaybreaks
\begin{align*}
\sum_i 4mr_i{(r_i-s_i)^2\over s_i^2}&\le 4m\sqrt{\sum_i {r_i^2\over s_i^2}}\sqrt{\sum_i {(r_i-s_i)^4\over s_i^2}}
&&\tag{Using Cauchy-Schwarz inequality}\\
&\le 4m\sqrt{\sqrt{K}\e[T]+2K} \sum_i {(r_i-s_i)^2\over s_i}\\
&\le 4\e[T] (K^{1/4}\sqrt{\e[T]}+\sqrt{2K}).
\end{align*}
\endgroup
Together we get
\begingroup
\allowdisplaybreaks
\begin{align*}
\var[T] &\le \sqrt{K}\e[T]+2K + 4\e[T] (K^{1/4}\sqrt{\e[T]}+\sqrt{2K})\\
&\le 2K + 7\sqrt{K}\e[T]+4K^{1/4}(\e[T])^{3/2}.
\end{align*}
\endgroup
\end{proof}
\end{document}